\documentclass[acmsmall, screen]{acmart}
\usepackage{pckgs_and_defs}
\usepackage{ulem}
\AtBeginDocument{%
  \providecommand\BibTeX{{%
    \normalfont B\kern-0.5em{\scshape i\kern-0.25em b}\kern-0.8em\TeX}}}

\setcopyright{acmcopyright}
\copyrightyear{2021}
\acmYear{2021}
\acmDOI{nn.nnnn/nnnnnnn.nnnnnnn}

\acmJournal{JACM}
\acmVolume{Unassigned}
\acmNumber{Unassigned}
\acmMonth{1}
\acmPrice{}

\begin{document}
\title{Learning-`N-Flying: A Learning-based, Decentralized Mission Aware UAS Collision Avoidance Scheme}
\author{Al\"ena Rodionova}
\email{alena.rodionova@seas.upenn.edu}
\orcid{0000-0001-8455-9917}
\affiliation{%
  \institution{University of Pennsylvania}
  \department{Department of Electrical and Systems Engineering}
  \city{Philadelphia}
  \state{PA}
  \postcode{19104}
  \country{USA}
}
\author{Yash Vardhan Pant}
\email{yashpant@berkeley.edu}
\affiliation{%
  \institution{University of California, Berkeley}
  \department{Department of Electrical Engineering and Computer Sciences}
  \city{Berkeley}
  \state{CA}
  \country{USA}
}

\author{Connor Kurtz}
\email{kurtzco@oregonstate.edu}
\affiliation{%
	\institution{Oregon State University}
	\department{School of Electrical Engineering and Computer Science}
	\city{Corvallis}
	\state{OR}
	\country{USA}
}

\author{Kuk Jang}
\email{jangkj@seas.upenn.edu}
\affiliation{%
  \institution{University of Pennsylvania}
  \department{Department of Electrical and Systems Engineering}
  \city{Philadelphia}
  \state{PA}
  \postcode{19104}
  \country{USA}
}

\author{Houssam Abbas}
\email{houssam.abbas@oregonstate.edu}
\affiliation{%
	\institution{Oregon State University}
	\department{School of Electrical Engineering and Computer Science}
	\city{Corvallis}
	\state{OR}
	\country{USA}
}

\author{Rahul Mangharam}
\email{rahulm@seas.upenn.edu}
\affiliation{%
	\institution{University of Pennsylvania}
	\department{Department of Electrical and Systems Engineering}
	\city{Philadelphia}
	\state{PA}
	\postcode{19104}
	\country{USA}
}


\begin{abstract}
Urban Air Mobility, the scenario where hundreds of manned and \ac{UAS} carry out a wide variety of missions (e.g. moving humans and goods within the city), is gaining acceptance as a transportation solution of the future. One of the key requirements for this to happen is safely managing the air traffic in these urban airspaces. Due to the expected density of the airspace, this requires fast autonomous solutions that can be deployed online. We propose Learning-`N-Flying (LNF) a multi-UAS Collision Avoidance (CA) framework. It is decentralized, works on-the-fly and allows autonomous \ac{UAS} managed by different operators to safely carry out complex missions, represented using Signal Temporal Logic, in a shared airspace. We initially formulate the problem of predictive collision avoidance for two \ac{UAS} as a mixed-integer linear program, and show that it is intractable to solve online. Instead, we first develop Learning-to-Fly (L2F) by combining: a) learning-based decision-making, and b) decentralized convex optimization-based control. LNF extends L2F to cases where there are more than two \ac{UAS} on a collision path. Through extensive simulations, we show that our method can run online (computation time in the order of milliseconds), and under certain assumptions has failure rates of less than $1\%$ in the worst-case, improving to near $0\%$ in more relaxed operations. We show the applicability of our scheme to a wide variety of settings through multiple case studies.	
\end{abstract}



\maketitle
\section{Introduction}
\label{sec:intro}

With the increasing footprint and density of metropolitan cities, there is a need for new transportation solutions that can move goods and people around rapidly and without further stressing road networks. \acf{UAM} \cite{NASA2018UAM} is a one such concept quickly gaining acceptance \cite{Booz2018} as a means to improve connectivity in metropolitan cities. In such a scenario, hundreds of Autonomous manned and \acf{UAS} will carry goods and people around the city, while also performing a host of other missions. A critical step towards making this a reality is safe traffic management of the all the \ac{UAS} in the airspace. Given the high expected UAS traffic density, as well as the short timescales of the flights, \acf{UTM} needs to be autonomous, and guarantee a high degree of safety, and graceful degradation in cases of overload. 
The first requirement for automated UTM is that its algorithms be able to accommodate a wide variety of missions, since the different operators have different goals and constraints. 
The second requirement is that as the number of UAS in the airspace increases, the runtimes of the UTM algorithms does not blow up - at least up to a point.
Thirdly, it must provide guaranteed collision avoidance in most use cases, and degrade gracefully otherwise; that is, the determination of whether it will be able to deconflict two UAS or not must happen sufficiently fast to alert a higher-level algorithm or a human operator, say, who can impose additional constraints. 

In this paper we introduce and demonstrate a new algorithm, LNF, for multi-UAS planning in urban airspace. 
LNF starts from multi-UAS missions expressed in Signal Temporal Logic (STL), a formal behavioral specification language that can express a wide variety of missions and supports automated reasoning.
In general, a mission will couple various UAS together through mutual separation constraints, and this coupling can cause an exponential blowup in computation.
To avoid this, LNF lets every UAS plan independently of others, while ignoring the mutual separation constraints.
This independent planning step is performed using Fly-by-Logic, our previous UAS motion planner.
An online collision avoidance procedure then handles potential collisions on an as-needed basis, i.e. when two UAS that are within communication range detect a future collision between their pre-planned trajectories. 
Even online optimal collision avoidance between two UAS requires solving a Mixed-Integer Linear Program (MILP). 
LNF avoids this by using a recurrent neural network which maps the current configuration of the two UAS to a sequence of discrete decisions. The network's inference step runs much faster (and its runtime is much more stable) than running a MILP solver.
The network is trained offline on solutions generated by solving the MILP. 
To generalize from two UAS collision avoidance to multi-UAS, we introduce another component to LNF: Fly-by-Logic generates trajectories that satisfy their STL missions, and a robustness tube around each trajectory. 
As long as the UAS is within its tube, it satisfies its mission.
To handle a collision between 3 or more UAS, LNF shrinks the robustness tube for each trajectory in such a way that sequential 2-UAS collision avoidance succeeds in deconflicting all the UAS.

We show that LNF is capable of successfully resolving collisions between UAS even within high-density airspaces and the short timescales, which are exactly the scenarios expected in UAM. LNF creates opportunities for safer UAS operations and therefore safer UAM.

\paragraph{Contributions of this work}
In this paper, we present an online, decentralized and mission-aware UAS \acf{CA} scheme that combines machine learning-based decision-making with Model Predictive Control (MPC). 
The main contributions of our approach are:
\begin{enumerate}
	\item It systematically combines machine learning-based decision-making\footnote{With the offline training and fast online application of the learned policy, see Sections~\ref{sec:learning_supervised} and \ref{sec:experiments_l2f}.} with an 
	MPC-based \ac{CA} controller. This allows us to decouple the usually hard-to-interpret machine learning component and the safety-critical low-level controller, and also \textit{repair} potentially unsafe decisions by the ML components. We also present a sufficient condition for our scheme to successfully perform \ac{CA}. 
	
	\item LNF collision avoidance avoids the live-lock condition where pair-wise \ac{CA} continually results in the creation of collisions between other pairs of \ac{UAS}.
	
	\item Our formulation is \textit{mission-aware}, i.e. \ac{CA} does not result in violation of the UAS mission. As shown in \cite{itsc20}, this also enables faster STL-based mission planning for a certain class of STL specifications. 
	
	\item Our approach is computationally lightweight with a computation time of the order of $10ms$ and can be used online.
	
	\item Through extensive simulations, we show that the worst-case failure rate of our method is less than $1\%$, which is a significant improvement over other approaches including \cite{itsc20}.
\end{enumerate}

\begin{figure}[t]
	\centering
	\includegraphics[width=0.45\textwidth]{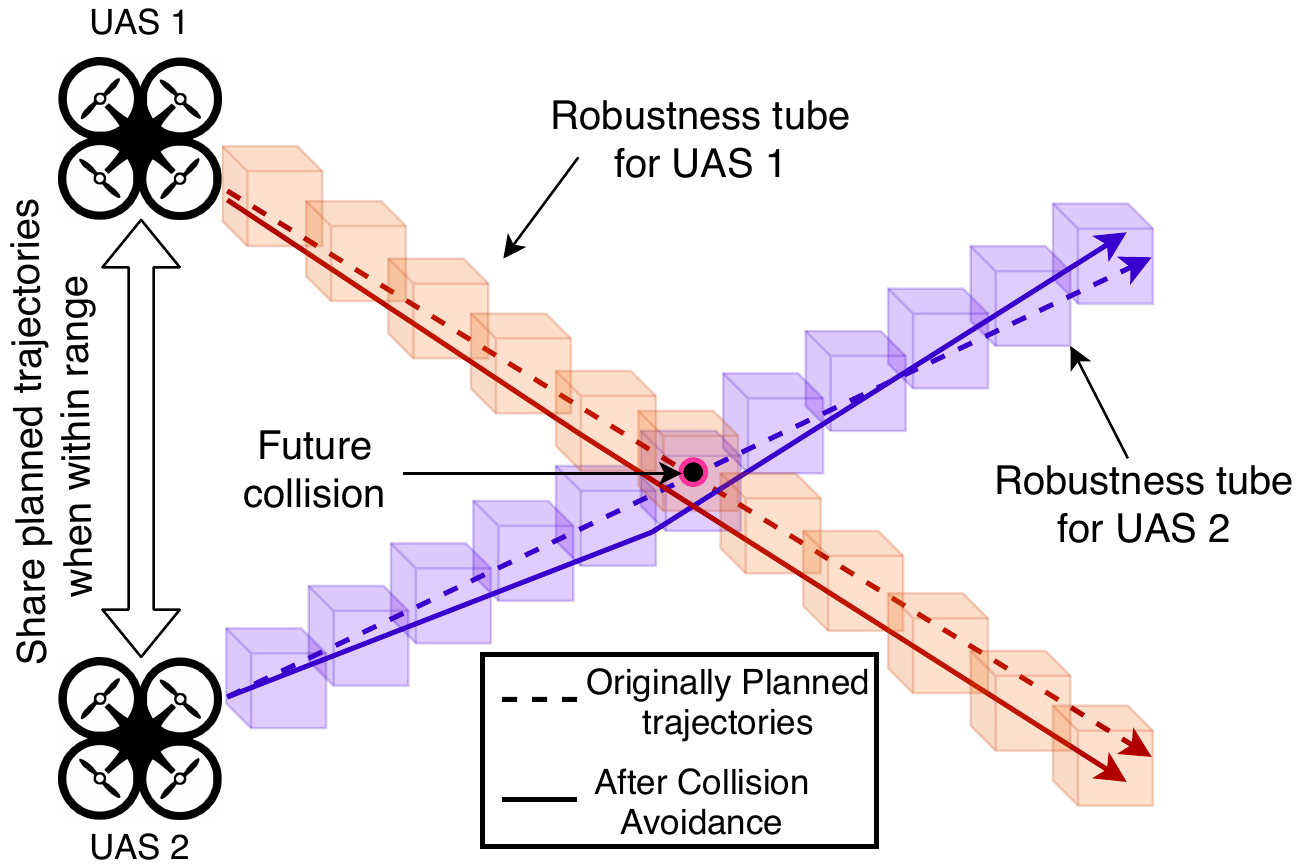}
	\vspace{-10pt}
	\caption{Two UAS communicating their planned trajectories, and cooperatively maneuvering within their \textit{robustness tubes} to avoid a potential collision in the future.}
	\label{fig:CAdiagram}
\end{figure}

\paragraph{Related Work.}
\textbf{UTM and Automatic Collision Avoidance approaches}
Collision avoidance (CA) is a critical component of UAS Traffic Management (UTM). The NASA/FAA Concept of Operations \cite{FAA2018UTM} and \cite{maxetal} present airspace allocation schemes where UAS are allocated airspace in the form of non-overlapping space-time polygons. Our approach is less restrictive and allows overlaps in the polygons, but performs online collision avoidance on an as-needed basis. A tree search-based planning approach for UAS CA is explored in \cite{UTMTCL4}. The next-gen CA system for manned aircrafts, ACAS-X \cite{kochenderfer2012next} is a learning-based approach that provides vertical separation recommendations. ACAS-Xu \cite{ACASXu} relies on a look-up table to provide high-level recommendations to two UAS.
It restricts desired maneuvers for CA to the vertical axis for cooperative traffic, and the horizontal axis for uncooperative traffic. 
While we consider only the cooperative case in this work, our method does not restrict \ac{CA} maneuvers to any single axis of motion. 
Finally, in its current form, ACAS-Xu also does not take into account any higher-level mission objectives, unlike our approach. This excludes its application to low-level flights in urban settings. The work in \cite{fabra2019distributed} presents a decentralized, mission aware CA scheme, but requires time of the order of seconds for the UAS to communicate and safely plan around each other, whereas our approach has a computation times in milliseconds.

\textbf{Multi-agent planning with temporal logic objectives}
Multi-agent planning for systems with temporal logic objectives has been well studied as a way of safe mission planning. Approaches for this usually rely on grid-based discretization of the workspace \cite{SahaRSJ14, DeCastro17}, or a simplified abstraction of the dynamics of the agents \cite{Drona,AksarayCDC16}.
\cite{MaICUAS16} combines a discrete planner with a continuous trajectory generator. 
Some methods \cite{4459804, 1582935, 1641832} work for subsets of Linear Temporal Logic (LTL) that do not allow for explicit timing bounds on the mission requirements.
The work in \cite{SahaRSJ14} allows some explicit timing constraints. However, it restricts motion to a discrete set of motion primitives. The predictive control method of \cite{Raman14_MPCSTL} uses the full STL grammar; it handles a continuous workspace and linear dynamics of robots, however its reliance on mixed-integer encoding (similar to \cite{Saha_acc16,KaramanF11_LTLrouting}) limit its practical use as seen in \cite{pant2017smooth}. The approach of \cite{pant2018fly} instead relies on optimizing a smooth (non-convex) function for generating trajectories for fleets of multi-rotor UAS with STL specifications. While these methods can ensure safe operation of multi-agent systems, these are all centralized approaches, i.e. require joint planning for all agents and do not scale well with the number of agents. In our framework, we use the planning method of \cite{pant2018fly}, but we let each UAS plan independently of each other in order for the planning to scale. We ensure the safe operation of all UAS in the airspace through the use of our predictive collision avoidance scheme.

\paragraph{Organization of the paper}
The rest of the paper is organized as follows.
Section~\ref{sec:preliminaries} covers preliminaries on Signal Temporal Logic and trajectory planning. 
In Section \ref{sec:CA} we formalize the two-\ac{UAS} \ac{CA} problems, state our main assumptions, and develop a baseline centralized solution via a MILP formulation. 
Section~\ref{sec:l2f} presents a decentralized learning-based collision avoidance framework for UAS pairs. 
In Section~\ref{sec:lnf} we extend this approach to support cases when \ac{CA} has to be performed for three or more UAS.
We evaluate our methods through extensive simulations, including three case studies in Section~\ref{sec:exp}.
Section~\ref{sec:conclusions} concludes the paper.

\section{Preliminaries: Signal Temporal Logic-based UAS planning}
\label{sec:preliminaries}

\paragraph{Notation.}
For a vector $x=(x_1,\ldots,x_m) \in \Re^m$, $\|x\|_\infty = \max_i |x_i|$.

\begin{figure}[t]
	\includegraphics[width=0.99\textwidth]{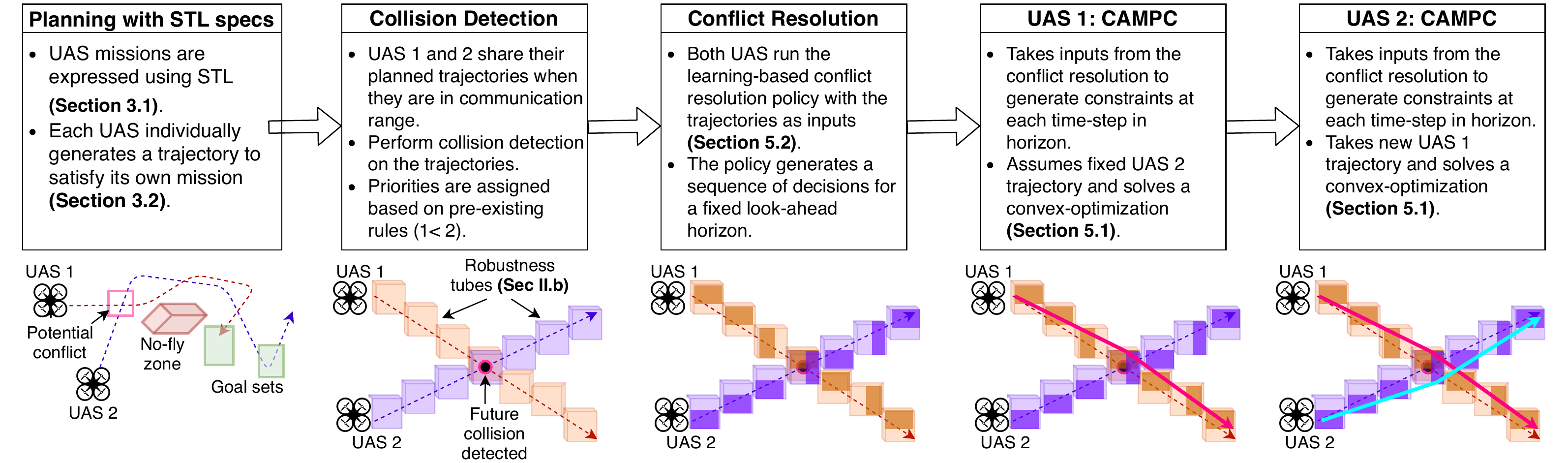}
	\vspace{-10pt}
	\caption{Step-wise explanation and visualization of the framework. Each UAS generates its own trajectories to satisfy a mission expressed as a Signal Temporal Logic (STL) specification, e.g. regions in green are regions of interest for the UAS to visit, and the no-fly zone corresponds to infrastructure that all the UAS must avoid. When executing these trajectories, UAS communicate their trajectories to others in range to detect any collisions that may happen in the near future. If a collision is detected, the two UAS execute a conflict resolution scheme that generates a set of additional constraints that the UAS must satisfy to avoid the collision. A co-operative CA-MPC controls the UAS to best satisfy these constraints while ensuring each UAS's STL specification is still satisfied. This results in new trajectories (in solid pink and blue) that will avoid the conflict and still stay within the predefined robustness tubes.}
	\label{fig:concept}
	\vspace{-10pt}
\end{figure}

\subsection{Introduction to Signal Temporal Logic and its Robustness}
\label{sec:STL_intro}
Let $\TDom = \{0,dt,2dt,3dt \ldots\}$ be a discrete time domain with sampling period $dt$ and let $\mathcal{X} \subset \Re^m$ be the state space.
A \textit{signal} is a function $\sstraj: E \rightarrow \mathcal{X}$ where $E \subseteq \TDom$; The $k^{\text{th}}$ element of $\sttraj$ is written $x_k$, $k\geq 0$. 
Let $\SigSpace$ be the set of all signals.

Signal specifications are expressed in Signal Temporal Logic (STL) \cite{MalerN2004STL}, of which we give an informal description here. 
An STL formula $\formula$ is created using the following grammar:
\[\formula  \defeq \top~|~p~|~\neg \formula ~|~ \formula_1 \lor \formula_2 ~|~ \eventually_{[a,b]} \formula ~|~ \always_{[a,b]} \formula ~|~ \formula_1 \until_{[a,b]} \formula_2\]
Here, $\top$ is logical True, $p$ is an atomic proposition, i.e. a basic statement about the state of the system, $\neg,\lor$ are the usual Boolean negation and disjunction, $\eventually$ is Eventually, $\always$ is Always and $\until$ is Until. 
It is possible to define the $\eventually$ and $\always$ in terms of Until $\until$, but we make them base operations because we will work extensively with them.

An STL specification $\varphi$ is interpreted over a signal, e.g. over the trajectories of quad-rotors, and evaluates to either \textit{True} or \textit{False}. 
For example, operator Eventually ($\eventually$) augmented with a time interval $\eventually_{[a,b]}\varphi$ states that $\varphi$ is $\textit{True}$ at some point within $[a,b]$ time units. Operator Always ($\always$) would correspond to $\varphi$ being $\textit{True}$ everywhere within time $[a,b]$.
The following example demonstrates how STL captures operational requirements for two UAS:
\begin{exmp}
\label{ex:reach_avoid_exmp}
\textit{(A two UAS timed reach-avoid problem)} 
Two quad-rotor UAS are tasked with a mission with spatial and temporal requirements in the workspace schematically shown in Figure~\ref{fig:concept}:

\begin{enumerate}

\item Each of the two UAS has to reach its corresponding $\text{Goal}$ set (shown in green) within a time of $6$ seconds after starting. 
UAS $j$ (where $j\in \{1,2\}$), with position denoted by $\mathbf{p}_j$, has to satisfy: $\varphi_{\text{reach}, j} = 
\eventually_{[0,6]} (\mathbf{p}_j \in \text{Goal}_j)$.
The \textit{Eventually} operator over the time interval $[0,6]$ requires UAS $j$ to be inside the set $\text{Goal}_j$ at some point within $6$ seconds. 

\item The two UAS also have an $\text{Unsafe}$ (in red) set to avoid, e.g. a no-fly zone. For each UAS $j$, this is encoded with \textit{Always} and \textit{Negation} operators:

$\varphi_{\text{avoid},j} = \always_{[0,6]} \neg (\mathbf{p}_j \in 
\text{Unsafe})$

\item Finally, the two UAS should be separated by at least $\delta$ meters along every axis of motion:

$\varphi_{\text{separation}} = \always_{[0,6]} ||\mathbf{p}_1 - \mathbf{p}_2||_{\infty} 
\geq \delta$

\end{enumerate}

The 2-UAS timed reach-avoid specification is thus:
\begin{equation}
\label{eq:timed_RA}
\varphi_{\text{reach-avoid}} = \bigwedge_{j=1}^2 ( \varphi_{\text{reach},j} \land 
\varphi_{\text{avoid},j}) \land \varphi_\text{separation}
\end{equation}
\end{exmp}

To satisfy $\varphi$ a planning method generates trajectories $\mathbf{p}_1$ and $\mathbf{p}_2$ of a duration at least $hrz(\varphi)= 6$s, where $hrz(\varphi)$ is the time \textit{horizon} of $\varphi$. 
If the trajectories satisfy the specification, i.e. $(\mathbf{p}_1,\, \mathbf{p}_2) \models \varphi$, then the specification $\varphi$ evaluates to \textit{True}, otherwise it is \textit{False}. 
In general, an upper bound for the time horizon can be computed as shown in \cite{Raman14_MPCSTL}. 
In this work, we consider specifications such that the horizon is bounded. More details on STL can be found in \cite{MalerN2004STL} or \cite{Raman14_MPCSTL}. 
In this paper, we consider discrete-time STL semantics which are defined over discrete-time trajectories.


%
%
%
%


The \textit{Robustness} value \cite{FainekosP09tcs} $\rho_\formula(\sttraj)$ of an STL formula $\formula$ with respect to the signal $\mathbf{x}$ is a real-valued function of $\mathbf{x}$ that has the important following property:

\begin{theorem} \cite{FainekosP09tcs}
	\label{thm:rob objective new}
	(i) For any $\sstraj \in \SigSpace$ and STL formula $\formula$, 
	if $\robf(\sstraj) <0$ then $\sstraj$ violates $\formula$, and if $\robf(\sstraj) > 0$ then $\sstraj$ satisfies $\formula$. 
	The case $\robf(\sstraj) =0$ is inconclusive.
	
	(ii) Given a discrete-time trajectory $\sstraj$ such that $\sstraj \models \formula$ with robustness value $\robf(\sttraj) = r>0$, then any trajectory $\mathbf{x}'$ that is within $r$ of $\sstraj$ at each time step, i.e. $||x_k-x'_k||_\infty < r, \, \forall k \in \mathbb{H}$, 
	is such that $\mathbf{x}' \models \formula$ (also satisfies $\formula$).
\end{theorem} 


\subsection{UAS planning with STL specifications}
\label{sec:problem_planning}

Fly-by-logic \cite{pant2017smooth,pant2018fly} generates trajectories by centrally planning for fleets of UAS with STL specifications, e.g. the specification $\varphi_{\textit{reach-avoid}}$ of example \ref{ex:reach_avoid_exmp}. 
It maximizes the robustness function by picking waypoints 
 for all UAS through a centralized, non-convex optimization.

While successful in planning for multiple multi-rotor UAS, performance degrades as the number of UAS increases, in particular because for $N$ UAS, $N \choose 2$ terms are needed for specifying the pair-wise separation constraint $\formula_{\textit{separation}}$.
For these reasons, the method cannot be used for real-time planning. 
In this work, we use the underlying optimization of  \cite{pant2018fly} to generate trajectories, but ignore the mutual separation requirement, allowing each UAS to independently (and in parallel) solve for their own STL specification. 
For the timed reach-avoid specification \eqref{eq:timed_RA} in example \ref{ex:reach_avoid_exmp}, this is equivalent to each UAS generating its own trajectory to satisfy $\varphi_j = \varphi_{\textit{reach}, j} \land \varphi_{\textit{avoid}, j}$, independently of the other UAS. 
Ignoring the collision avoidance requirement $\varphi_\textit{separation}$ in the planning stage allows for the specification of \eqref{eq:timed_RA} to be decoupled across UAS. 
Therefore, this approach requires \textit{online} UAS collision avoidance. This is covered in the following section.

\section{Problem formulation: Mission aware UAS Collision Avoidance} 
\label{sec:CA}



We consider the case where two UAS flying pre-planned trajectories are required to perform collision avoidance if their trajectories are on path for a \textit{conflict}.

\begin{mydef}
	\label{def:msep}
	\textbf{2-UAS Conflict}: Two UAS, with discrete-time positions $\mathbf{p}_1$ and $\mathbf{p}_2$ are said to be in \textit{conflict} at time step $k$ if $||p_{1,k}-p_{2,k}||_\infty < \delta$, where $\delta$ is a predefined minimum separation distance\footnote{A more general polyhedral constraint of the form $M(p_{1,k}-p_{2,k})< q$ can be used for defining the conflict.}. Here, $p_{j,k}$ represents the position of UAS $j$ at time step $k$. 
\end{mydef}

While flying their independently planned trajectories, two UAS that are within communication range share an $H$-step look-ahead of their trajectories and check for a potential conflict in those $H$ steps.
We assume the UAS can communicate with each other in a manner that allows for enough advance notice for avoiding collisions, e.g. using 5G technology. While the details of this are beyond the scope of this paper, we formalize this assumption as follows:

\begin{myass}
	\label{assumption0}
	The two UAS in conflict have a communication range that is at least greater than their $n$-step forward reachable set \cite{dahleh2004lectures} ($n\geq 1$) \footnote{This set can be computed offline as we know the dynamics and actuation limits for each UAS.}. That is, the two UAS will not collide immediately in at least the next $n$-time steps, enabling them to communicate with each other to avoid a collision. Here $n$ is potentially dependent on the communication technology being used.  	
\end{myass}


\begin{mydef}
	\label{def:robustnesstube}
	\textbf{Robustness tube}: Given an STL formula $\varphi$ and a discrete-time position trajectory $\mathbf{p}_j$ that satisfies $\varphi$ (with associated robustness $\rho$), the (discrete) \textit{robustness tube} around $\mathbf{p}_j$ is given by $\mathbf{P}_j = \mathbf{p}_j\oplus \mathbb{B}_{\rho}$, 
	where $\mathbb{B}_\rho$ is a 3D cube with sides $2\rho$ and $\oplus$ is the Minkowski sum operation ($A\oplus B \defeq \{a+b\such a\in A,b\in B\}$). We say the \textit{radius} of this tube is $\rho$ (in the inf-norm sense). 
\end{mydef}
Robustness tube defines the space around the UAS trajectory, such that as long as the UAS stays within its robustness tube, it will satisfy STL specification for which it was generated.
See examples of the robustness tubes in Figures~\ref{fig:CAdiagram} and \ref{fig:concept}. 

%


The following assumption relates the minimum allowable radius $\rho$ of the robustness tube to the minimum allowable separation $\delta$ between two UAS.

\begin{myass}
	\label{assumption1}
	For each of the two UAS in conflict, the radius of the robustness tube is greater than $\delta/2$, i.e. $\min (\rho_1,\rho_2) \geq \delta/2$ where $\rho_1$ and $\rho_2$ are the robustness of UAS 1 and 2, respectively.
\end{myass}
This assumption defines the case where the radius of the robustness tube is wide enough to have two UAS placed along opposing edges of their respective tubes and still achieve the minimum separation between them. We assume that all the trajectories generated by the independent planning have sufficient robustness to satisfy this assumption (see Sec. \ref{sec:problem_planning}).
Now we define the problem of collision avoidance with satisfaction of STL specifications:
\begin{myprob}
	\label{prob:deconfliction}
	Given two planned $H$-step UAS trajectories $\mathbf{p}_1$ and $\mathbf{p}_2$ that have a conflict, the collision avoidance problem is to find a new sequence of positions $\mathbf{p}_1'$ and $\mathbf{p}_2'$ that meet the following conditions:
	\begin{subequations}
		\begin{align}
		||p_{1,k}'-p_{2,k}'|| \geq \delta, \, &\forall k \in \{0,\dotsc,H\} \label{eq:msep}\\
		p_{j,k}' \in P_{j,k}, \, &\forall k \in \{0,\dotsc,H\},\,\forall j\in\{1,2\}. \label{eq:intube}
		\end{align}
	\end{subequations}
\end{myprob}
That is, we need a new trajectory for each UAS such that they achieve minimum separation distance and also stay within the robustness tube around their originally planned trajectories. 




\paragraph{Convex constraints for collision avoidance}

Let $z_k=p_{1,k} - p_{2,k} $ be the difference in UAS positions at time step $k$. For two UAS not to be in conflict, we need 
\begin{equation}
\label{eq:noconf}
z_k \not \in \mathbb{B}_{\delta/2},\ \forall k\in\{0,\ldots,H\},
\end{equation}
This is a non-convex constraint. For a computationally tractable controller formulation which solves Problem \ref{prob:deconfliction}, we define convex constraints that when satisfied imply Equation \eqref{eq:noconf}. 
The $3$D cube $\mathbb{B}_{\delta/2}$ can be defined by a set of linear inequality constraints of the form $\widetilde{M}^i z \leq \widetilde{q}^i, \, \forall i\in\{1,\ldots,6\}$. Equation~\eqref{eq:noconf} is satisfied when $\exists i \, | \widetilde{M}^i z > \widetilde{q}_i$. Let $M = -\widetilde{M}$ and $q = -\widetilde{q}$, then $\forall i \in \{1,\ldots,6\}$, 
\begin{equation}
\label{eq:pickaside}
M^i(p_{1,k}-p_{2,k}) < {q}^i \Rightarrow (p_{1,k}-p_{2,k}) \not \in \mathbb{B}_{\delta/2}
\end{equation}

Intuitively, picking one $i$ at time step $k$ results in a configuration (in position space) where the two UAS are separated in one of two ways along one of three axes of motion\footnote{Two ways along one of three axes defines $6$ options, $i\in\{1,\ldots,6\}$.}. For example, if at time step $k$ we select $i$ with corresponding $M^i=[0, 0, 1]$ and $q^i = -\delta$, it implies that UAS 2 flies over UAS 1 by $\delta$ meters, and so on.

\paragraph{A Centralized solution via a MILP formulation}

Here, we formulate a \ac{MILP} to solve the two UAS \ac{CA} problem of problem \ref{prob:deconfliction} in a predictive, receding horizon manner. For the formulation, we consider a $H$-step look ahead that contains the time steps where the two UAS are in conflict. Let the dynamics of either UAS\footnote{For simplicity we assume both UAS have identical dynamics associated with multi-rotor robots, however our approach would work otherwise.} be of the form $x_{k+1} = Ax_k + Bu_k$. 
At each time step $k$, the UAS state is defined as $x_k =[p_k,\,v_k]^T\in \mathbb{R}^6$, where $p$ and $v$ are the UAS positions and velocities in the 3D space. Let $C$ be the observation matrix such that $p_k=Cx_k$. The inputs $u_k \in \mathbb{R}^3$ are the thrust, roll and pitch of the UAS. The matrices $A$ and $B$ are obtained through linearization 
of the UAS dynamics around hover and discretization in time, see \cite{luukkonen2011modelling} and \cite{pant2015co} for
more details. 
Let $\mathbf{x}_j \in \mathbb{R}^{6(H+1)}$ be the pre-planned full state trajectories, $\mathbf{x}_j' \in \mathbb{R}^{6(H+1)}$ the new full state trajectories and $\mathbf{u}_j' \in \mathbb{R}^{3H}$ the new controls to be computed for the UAS $j=1,2$. 
Let $\mathbf{b} \in \{0,1\}^{6(H+1)}$ be binary decision variables, and $\mu$ is a large positive number, then the MILP problem is defined as:
\begin{equation}
\label{eq:CentralMILP}
	\begin{aligned}
	&\min_{\mathbf{u}_1', \, \mathbf{u}_2', \, \mathbf{b}} J(\mathbf{x}_1', \, \mathbf{u}_1', \, \mathbf{x}_2', \, \mathbf{u}_2') \\
	x_{j,0}' &= x_{j,0}, \, \forall j \in \{1,2\} \\
	x_{j,k+1}' &= Ax_{j,k}' + Bu_{j,k}', \, \forall k \in \{0,\dotsc,H-1\} , \, \forall j \in \{1,2\}\\
	Cx'_{j,k} &\in P_{j,k}, \, \forall k \in \{0,\dotsc,H\} , \, \forall j \in \{1,2\}\\
	M^{i}C\,(x_{1,k}'-x_{2,k}') &\leq {q}_{i} +\mu(1-b^i_{k}), \, \forall k \in \{0,\dotsc,H\}, \forall i \in \{1,\dotsc,6\} \\
	\sum_{i=1}^6 b^i_{k} &\geq 1, \, \forall k \in \{0,\dotsc,H\} \\
	u_{j,k}' &\in U, \, \forall k \in \{0,\dotsc,H\} , \, \forall j\in \{1,2\}\\
	x_{j,k}' &\in X, \, \forall k \in \{0,\dotsc,H+1\}, \forall j\in \{1,2\}.
	\end{aligned}
\end{equation}
Here $b^i_k$ encodes action $i=1,\dotsc,6$ taken for avoiding a collision at time step $k$ which corresponds to a particular side of the cube $\mathbb{B}_{\delta/2}$.
Function $J$ could be any cost function of interest, we use $J=0$ to turn \eqref{eq:CentralMILP} into a feasibility problem.
A solution (when it exists) to this MILP results in new trajectories ($\mathbf{p}_1', \, \mathbf{p}_2'$) that avoid collisions and stay within their respective robustness tubes of the original trajectories, and hence are a solution to problem \ref{prob:deconfliction}.

Such optimization is joint over both UAS. It is impractical as it would either require one UAS to solve for both or each UAS to solve an identical optimization that would also give information about the control sequence of the other UAS. Solving this MILP in an online manner is also intractable, as we shown in Section \ref{sec:exp_results}.


\section{Learning-2-Fly: Decentralized Collision Avoidance for UAS pairs}
\label{sec:l2f}

To solve problem \ref{prob:deconfliction} in an online and decentralized manner, we develop our framework, Learning-to-Fly (L2F). Given a predefined priority among the two UAS, this combines a learning-based \ac{CR} scheme (running aboard each UAS) that gives us the discrete components of the \ac{MILP} formulation \eqref{eq:CentralMILP}, and a co-operative collision avoidance MPC for each UAS to control them in a decentralized manner.
We assume that the two UAS can communicate their pre-planned $N$-step trajectories $\mathbf{p}_1,\, \mathbf{p}_2$ to each other (refer to Sec. \ref{sec:problem_planning}), and then L2F solves problem \ref{prob:deconfliction} by following these steps (also see  Algorithm~\ref{alg:l2f}) : 

\begin{enumerate}
	\item \textbf{Conflict resolution:} UAS 1 and 2 make a \textit{sequence of decisions}, $\mathbf{d}=(d_0,\ldots,d_H)$ to avoid collision. Each $d_k\in\{1,\ldots\,6\}$ 
	represents a particular choice of $M$ and $q$ at time step $k$, see eq.~\eqref{eq:pickaside}.
	Section \ref{sec:learning_supervised} will describe our proposed learning-based method for picking $d_k$.
	\item \textbf{UAS 1 CA-MPC:} UAS 1 takes the conflict resolution sequence $\mathbf{d}$ from step 1 and solves a convex optimization to try to deconflict while assuming UAS 2 maintains its original trajectory. After the optimization the new trajectory for UAS 1 is sent to UAS 2.
	\item \textbf{UAS 2 CA-MPC:} (If needed) UAS 2 takes the same conflict resolution sequence $\mathbf{d}$ from step 1 and solves a convex optimization to try to avoid UAS 1's new trajectory. 
	Section~\ref{sec:ca_mpc} provides more details on CA-MPC steps 2 and 3.
\end{enumerate}

The visualization of the above steps is presented in Figure~\ref{fig:concept}.
Such decentralized approach differs from the centralized MILP approach, where both the binary decision variables and continuous control variables for each UAS are decided concurrently. 

\vspace{-5pt}
\subsection{Distributed and co-operative Collision Avoidance MPC (CA-MPC)}
\label{sec:ca_mpc}

%
Let $\mathbf{x}_j$ be the pre-planned trajectory of UAS $j$, $\mathbf{x}_\textit{avoid}$ be the pre-planned trajectory of the other UAS to which $j$ must attain a minimum separation, 
and let $prty_j \in \{-1, +1\}$ be the priority of UAS $j$.
Assume a decision sequence $\mathbf{d}$ is given: at each $k$ in the collision avoidance horizon, the UAS are to avoid each other by respecting \eqref{eq:pickaside}, namely $M^{d_k}(p_{1,k}-p_{2,k}) < {q}^{d_k}$.
Then each UAS $j=1,2$ solves the following Collision-Avoidance MPC optimization (CA-MPC):
\textbf{$\text{CA-MPC}_j(\mathbf{x}_j,\, \mathbf{x}_{avoid},\,\mathbf{P}_j,\ \mathbf{d},\, prty_j)$}:
\vspace{-5pt}
\begin{equation}
\label{eq:campc}
\begin{aligned}
& \min_{\mathbf{u}_j',\boldsymbol{\lambda}_j} \sum_{k=0}^H \lambda_{j,k} \\
x_{j,0}' &= x_{j,0} \\
x_{j,k+1}' &= Ax_{j,k}' + Bu_{j,k}', \, \forall k\in\{0,\dotsc,H-1\} \\
Cx_{j,k}' &\in P_{j,k}, \, \forall k\in\{0,\dotsc,H\} \\
prty_j\cdot M^{d_k}C\,(x_{avoid,k}-x_{j,k}') &\leq q^{d_k} + \lambda_{j,k},\, \forall k\in\{0,\dotsc,H\} \\
\lambda_{j,k} &\geq 0,\, \forall k\in\{0,\dotsc,H\} \\
u_{j,k}' &\in U,\, \forall k\in\{0,\dotsc,H \} \\
x_{j,k}' &\in X, \, \forall k \in \{0,\dotsc,H+1\}.
\end{aligned}
\end{equation}

This MPC optimization tries to find a new trajectory $\mathbf{x}_j'$ for the UAS $j$ that minimizes the slack variables $\lambda_{j,k}$ that correspond to violations in the minimum separation constraint $\eqref{eq:pickaside}$ w.r.t the pre-planned trajectory $\mathbf{x}_\textit{avoid}$ of the UAS in conflict. 
The constraints in \eqref{eq:campc} ensure that UAS $j$ respects its dynamics, input constraints, and state constraints to stay inside the robustness tube. 
An objective of $0$ implies that UAS $j$'s new trajectory satisfies the minimum separation between the two UAS, see Equation~\eqref{eq:pickaside}\footnote{Enforcing the separation constraint at each time step can lead to a restrictive formulation, especially in cases where the two UAS are only briefly close to each other. This does however give us an optimization with a structure that does not change over time, and can avoid collisions in cases where the UAS could run across each other more than once in quick succession (e.g. \url{https://tinyurl.com/arc-case}), which is something ACAS-Xu was not designed for.}.

\textbf{CA-MPC optimization for UAS 1:}
UAS 1, with lower priority, $prty_1 = -1$, first attempts to resolve the conflict for the given sequence of decisions $\mathbf{d}$:
\begin{align}
(\mathbf{x_1'}, \mathbf{u}_1', \boldsymbol{\lambda}_1)&=\textbf{CA-MPC}_1(\mathbf{x}_1,\mathbf{x}_2, \mathbf{P}_1, \mathbf{d}, -1)
\label{eq:drone1mpc}
\end{align}
An objective of $0$ implies that UAS 1 alone can satisfy the minimum separation between the two UAS. 
Otherwise, UAS 1 alone could not create separation and UAS 2 now needs to maneuver as well.

\textbf{CA-MPC optimization for UAS 2:}
If UAS 1 is unsuccessful at collision avoidance, UAS 1 communicates its current revised trajectory $\mathbf{x}_1'$ to UAS 2, with $prty_2 = +1$.
UAS 2 then creates a new trajectory $\mathbf{x}_2'$ 
(w.r.t the same decision sequence $\mathbf{d}$):
\begin{align}
(\mathbf{x}_2', \mathbf{u}_2', \boldsymbol{\lambda}_2)&=\textbf{CA-MPC}_2(\mathbf{x}_2,\mathbf{x}_1', \mathbf{P}_2, \mathbf{d}, +1)
\label{eq:drone2mpc}
\end{align}

Algorithm~\ref{alg:l2f} is designed to be computationally lighter than the 
MILP approach \eqref{eq:CentralMILP}, but unlike the MILP it is not complete.

\begin{algorithm}[t!]
	\SetAlgoLined
	\Notation{$(\mathbf{x}'_1, \mathbf{x}'_2,\mathbf{u}_1', \mathbf{u}_2')=\textbf{L2F}(\mathbf{x}_1, \mathbf{x}_2, \rtube_1,\rtube_2)$}
	\KwIn{Pre-planned trajectories $\mathbf{x}_1$, $\mathbf{x}_2$, robustness tubes $\rtube_1$, $\rtube_2$}
	\KwOut{Sequence of control signals $\mathbf{u}_1'$, $\mathbf{u}_2'$ for the two UAS, updated trajectories $\mathbf{x}_1'$, $\mathbf{x}_2'$}
	
	{Get $\mathbf{d}$ from conflict resolution}
	
	{UAS 1 solves CA-MPC optimization \eqref{eq:campc}:
	$(\mathbf{x}_1', \mathbf{u}_1', \boldsymbol{\lambda}_1)=\textbf{CA-MPC}_1(\mathbf{x}_1,\mathbf{x}_2, \mathbf{P}_1, \mathbf{d}, -1)$}
	
	\eIf{$\sum_k \lambda_{1,k} = 0$}
	{\textbf{Done}: UAS 1 alone has created separation; Set $\mathbf{u}_2'=\mathbf{u}_2$
	}
	{ UAS 1 transmits solution to UAS 2
		
		{UAS 2 solves CA-MPC optimization \eqref{eq:campc}:
		$(\mathbf{x}_2', \mathbf{u}_2', \boldsymbol{\lambda}_2)=\textbf{CA-MPC}_2(\mathbf{x}_2,\mathbf{x}_1', \mathbf{P}_2, \mathbf{d}, +1)$}
		
		\eIf{$\sum_k \lambda_{2,k} = 0$}{\textbf{Done:} UAS 2 has created 
			separation }
		{\eIf{$||p_{1,k}'-p_{2,k}'|| \geq \delta, \, \forall k = 0,\dotsc,H$}
			{\textbf{Done}: UAS 1 and UAS 2 created separation}
			{\textbf{Not done}: UAS still violate Equation \eqref{eq:msep}}}
	}
	
	Apply control signals $\mathbf{u}_1'$, $\mathbf{u}_2'$ if \textbf{Done}; else \textbf{Fail}.
	\caption{Learning-to-Fly: Decentralized and cooperative collision avoidance for two UAS. Also see Figure \ref{fig:concept}.}
	\label{alg:l2f}
	
\end{algorithm}
The solution of CA-MPC can be defined as follows:
\begin{definition}[Zero-slack solution]
	\label{def:zero_slack}
	The solution of the CA-MPC optimization \eqref{eq:campc}, is called the \textit{zero-slack solution} if for a given decision sequence $\mathbf{d}$ either
	
	1) there exists an optimal solution of \eqref{eq:campc} such that  $\sum_k\lambda_{1,k}=0$ or
	
	2) problem \eqref{eq:campc} is feasible with $\sum_k\lambda_{1,k}>0$ and there exists an optimal solution of \eqref{eq:campc} such that  $\sum_k\lambda_{2,k}=0$.
\end{definition}

The following Theorem~\ref{th:CAMPC_success} defines the sufficient condition for CA and Theorem~\ref{th:MILP_CAMPC_relation} makes important connections between the slack variables in CA-MPC formulation and binary variables in MILP. 
Both theorems
are direct consequences of the construction of CA-MPC optimizations. We omit the proofs for brevity.

\begin{theorem}[Sufficient condition for CA]
	\label{th:CAMPC_success}
	Zero-slack solution of \eqref{eq:campc}
	implies that the resulting trajectories for two UAS are non-conflicting and within the robustness tubes of the initial trajectories\footnote{Theorem~\ref{th:CAMPC_success} formulates a conservative result as \eqref{eq:pickaside} is a convex under approximation of the originally non-convex collision avoidance constraint \eqref{eq:noconf}. Indeed, non-zero slack $\exists k| \lambda_{2,k}>0$ does not necessarily imply the violation of the mutual separation requirement \eqref{eq:msep}. The control signals $u_1',u_2'$ computed by Algorithm~\ref{alg:l2f} can therefore in some instances still create separation between UAS even when the conditions of Theorem \ref{th:CAMPC_success} are not satisfied.}.
\end{theorem}


\begin{theorem}[Existence of the zero-slack solution]
	\label{th:MILP_CAMPC_relation}
	Feasibility of the MILP problem~\eqref{eq:CentralMILP} implies the existence of the zero-slack solution of CA-MPC optimization \eqref{eq:campc}.
\end{theorem}
The Theorem~\ref{th:MILP_CAMPC_relation} states that the binary decision variables $b^i_k$ selected by the feasible solution of the MILP problem \eqref{eq:CentralMILP}, when used to select the constraints (defined by $M,\,q$) for the CA-MPC formulations for UAS 1 and 2, imply the existence of a zero-slack solution of \eqref{eq:campc}.



\subsection{Learning-based conflict resolution}
\label{sec:learning_supervised}

Motivated by Theorem~\ref{th:MILP_CAMPC_relation},
we propose to learn offline the conflict resolution policy 
from the MILP solutions and then online use already learned policy.
To do so, we use a \textit{Long Short-Term Memory} (LSTM)~\cite{hochreiter1997long} recurrent neural network augmented with fully-connected layers.
LSTMs perform better than traditional recurrent neural networks on sequential prediction tasks~\cite{gers2002learning}.
 
 \begin{figure}[tb]
 	\begin{center}
 		\includegraphics[width=0.9\textwidth]{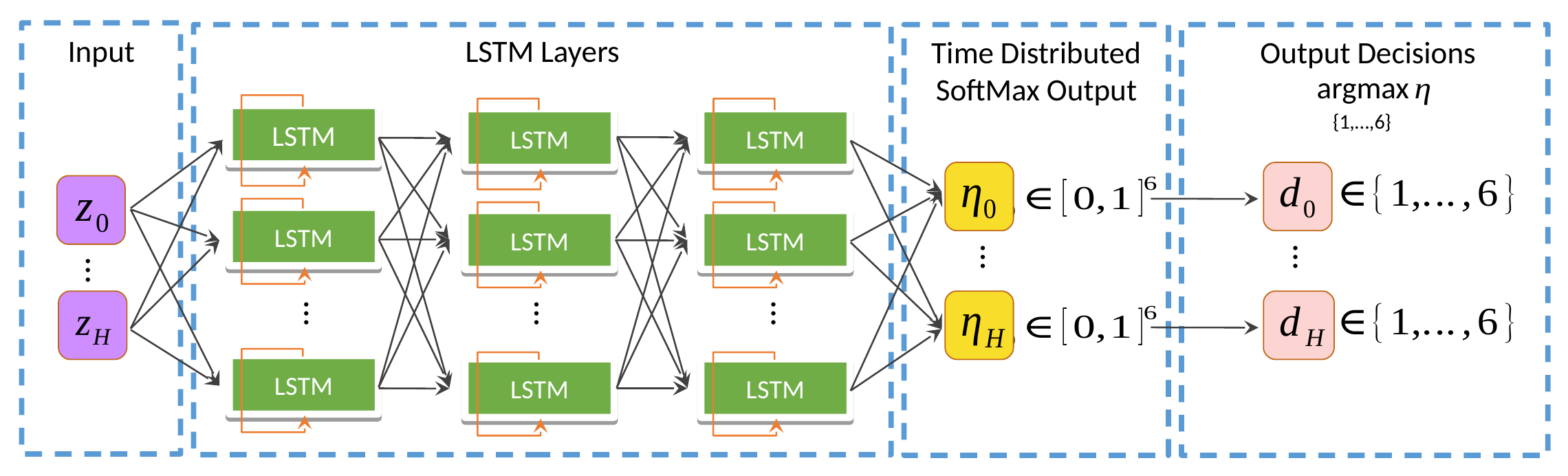}
 	\end{center}
 	\vspace{-10pt}
 	{\caption{Proposed LSTM model architecture for CR-S. LSTM layers are shown unrolled over $H$ time steps. The inputs are 
 			$z_k$ which are the differences between the planned UAS positions, and the outputs are decisions $d_k$ for conflict resolution at each time $k$ in the horizon.}
 	\label{fig:lstm_arc}}
 \end{figure} 

The network is trained to map a difference trajectory $\mathbf{z}=\mathbf{x}_1-\mathbf{x}_2$ (as in Equation~\eqref{eq:noconf}) to a decision sequence $\mathbf{d}$ that deconflicts pre-planned trajectories $\mathbf{x}_1$ and $\mathbf{x}_2$. For creating the training set, $\mathbf{d}$ is produced by solving the MILP problem~\eqref{eq:CentralMILP}, i.e. obtaining 
a sequence of binary decision variables $\mathbf{b}\in\{0,1\}^{6(H+1)}$ and translating it into the decision sequence $\mathbf{d}\in\{1,\ldots,6\}^{H+1}$. 


The proposed architecture is presented in 
 Figure~\ref{fig:lstm_arc}.
The input layer is 
connected to the block of three stacked LSTM layers.
The output layer is a time distributed dense layer with a 
softmax activation function that produces the class probability estimate $\eta_k=[\eta_k^1,\ldots,\eta_k^6]^\top$ for each $k\in\{0,\ldots,H\}$, which corresponds to a 
decision $d_k=\argmax_{i=1,\ldots 6} \eta_k^i$. 

\subsection{Conflict Resolution Repairing}
\label{sec:repair}

The total number of possible conflict resolution (CR) decision sequences of over a time horizon of $H$ steps is $H^6$. 
Learning-based collision resolution produces only one such CR sequence, and since it is not guaranteed to be correct, an inadequate CR sequence might lead to the CA-MPC being unable find a feasible solution of \eqref{eq:campc}, i.e. a failure in resolving a collision. To make the CA algorithm more resilient to such failures, we propose a heuristic that instead of generating only one CR sequence, generates a number of slightly modified sequences, aka backups, with an intention of increasing the probability of finding an overall solution for CA.     
We call it a \textit{CR repairing algorithm}. We propose the following scheme for CR repairing. 


\subsubsection{Na\"ive repairing scheme for generating CR decision sequences}
The na\"ive-repairing algorithm is based on the initial supervised-learning CR architecture, see Section~\ref{sec:learning_supervised}.
The proposed DNN model for CR has the output layer with a
softmax activation function that produces the class probability estimates $\eta_k=[\eta_k^1,\ldots,\eta_k^6]^\top$ for each time step $k$, see Figure~\ref{fig:lstm_arc}. Discrete decisions were chosen as:
\begin{equation}
d_k = \underset{i=1,\ldots 6}{\argmax}\  \eta_k^i,
\end{equation}
which corresponds to the highest probability class for time step $k$. Denote such choice of $d_k$ as $d_k^1$.

Analogously to the idea of top-1 and top-$S$ accuracy rates used in image classification~\cite{ILSVRC15}, where not only the highest predicted class counts but also the top $S$ most likely labels, 
we define higher order decisions $d^s_k$ as following:
instead of choosing the highest probability class at time step $k$,
one could choose the second highest probability class ($s=2$), third highest ($s=3$), up to the sixth highest ($s=6$).

Formally, the second highest probability class choice $d_k^2$ is defined as:
\begin{equation}
	d_k^2=\underset{i=1,\ldots 6,\,i\not=d_k^1}{\argmax}\  \eta_k^i
\end{equation}
In the same manner, we define decisions up to $d_k^6$. 
General formula for the $s$-th highest probability class, decision $d_k^s$ is defined as following ($s=1,\ldots,6$):
\begin{equation}
\label{eq:d_ks}
d_k^s=\underset{i=1,\ldots 6,\,i\not=d_k^j\ \forall j<s}{\argmax}\  \eta_k^i
\end{equation}

Using equation \eqref{eq:d_ks} to generate decisions $d_k$ at time step $k$, we define the na\"ive scheme for generating new decision sequences $\mathbf{d}'$ following Algorithm~\ref{alg:naive_rep}.

\begin{algorithm}
	\SetAlgoLined
	\Notation{$(\mathbf{x}_1', \mathbf{x}_2', \mathbf{u}_1', \mathbf{u}_2') = \textbf{Repairing}(\mathbf{x}_1, \mathbf{x}_2, \rtube_1, \rtube_2, \varUpsilon)$}
	\KwIn{Pre-planned trajectories $\mathbf{x}_1$, $\mathbf{x}_2$, robustness tubes $\mathbf{P}_1$, $\mathbf{P}_2$, original decision sequence $\mathbf{d}$, class probability estimates $\mathbf{\eta}$,
		set of collision indices: $\varUpsilon=\{ k:\ ||p_{1,k}'-p_{2,k}'|| < \delta, \ 0\leq k\leq H \}$.}
	\KwOut{Sequence of control signals $\mathbf{u}_1'$, $\mathbf{u}_2'$ for the two UAS, updated trajectories $\mathbf{x}_1'$, $\mathbf{x}_2'$}

	\For{$s=2,\ldots,6$}{
		
		Define repaired sequence $\mathbf{d}'$ using na\"ive scheme as follows:
		
		\begin{itemize}
			\item[-] $\forall k\not\in \varUpsilon:\  d'_k=d_k$
			\item[-] $\forall k\in \varUpsilon:\ d'_k=d_k^s=\argmax_{i=1,\ldots 6,\,i\not=d_k^j\ \forall j<s}\  \eta_k^i$
		\end{itemize}		
		
		$(\mathbf{x}_1', \mathbf{x}_2', \mathbf{u}_1', \mathbf{u}_2')=\textbf{CA-MPC}(\mathbf{x}_1, \mathbf{x}_2, \mathbf{P}_1, \mathbf{P}_2, \mathbf{d}')$
		
		\If{$||p_{1,k}'-p_{2,k}'|| \geq \delta, \, \forall k = 0,\dotsc,N$}
		{\textbf{Break}: Repaired CR sequence $\mathbf{d}'$ led to UAS 1 and UAS 2 creating separation
		}		
	} 
\If{$||p_{1,k}'-p_{2,k}'|| \geq \delta, \, \forall k = 0,\dotsc,H$}
{$\mathbf{d}'=\mathbf{d}$: Repairing failed. Return trajectories for the original decision sequence.} 
	\caption{Na\"ive scheme for CR repairing}
	\label{alg:naive_rep}
\end{algorithm}

\begin{example}
	Let the horizon of interest be only $H=5$ time steps and the initially obtained decision sequence be $\mathbf{d}=(1, 1, 1, 1, 1)$.
	Given the collision was detected at time steps 2 and 3, i.e. $\varUpsilon=(2, 3)$, let the second-highest probability decisions be
	$d_{2}^2=3$ and $d_{3}^2=5$.
	Then the proposed repaired decision sequence is
	$\mathbf{d}'=(1, 1, 3, 5, 1)$. If such CR sequence $\mathbf{d}'$ still violates the mutual separation requirement, then the na\"ive repairing scheme will propose another decision sequence using the third-highest probability decisions $d_3$. Let $d_2^3=2$ and $d_3^3=3$ then
	$\mathbf{d}'=(1, 1, 2, 3, 1)$. If it fails again, the next generated  sequence will use fourth-highest decisions, and so on up to the fifth iteration of the algorithm (requires $d_k^6$ estimates). If none of the sequences managed to create separation, the original CR sequence $\mathbf{d}=(1, 1, 1, 1, 1)$ will be returned. 
\end{example}

Other variations of the na\"ive scheme are possible. For example, one can use augmented set of collision indices $\varUpsilon$ or another order of decisions $d_k$ across the time indices, e.g. replace decisions $d_k$ one-by-one rather than all $d_k$ for collision indices $\varUpsilon$ at once. 
Moreover, other CR repairing schemes can be efficient and should be explored. We leave it for future work.

\section{Learning-`N-Flying: Decentralized Collision Avoidance for Multi-UAS Fleets}
\label{sec:lnf}

The L2F framework of Section ~\ref{sec:l2f} was tailored for \ac{CA} between two \ac{UAS}. When more than two \ac{UAS} are simultaneously on a collision path, applying L2F pairwise for all \ac{UAS} involved might not necessarily result in all future collisions being resolved. Consider the following example:

\begin{example}
	\label{ex:rts}
	Figure~\ref{fig:test-pairwise} depicts an experimental setup. Scenario consists of 3 UAS which must reach
	desired goal states within 4 seconds while avoiding each other, minimum allowed separation is set to $\delta=0.1m$.
	Initially pre-planned UAS trajectories have a simultaneous collision  across all UAS located at $(0,0,0)$.
	Robustness tubes radii were fixed at $\rho=0.055$ and
	UAS priorities were set in the increasing order, e.g. UAS with a lower index had a lower priority: $1<2<3$.
	First application of L2F 	
	lead to resolving collision for UAS 1 and UAS 2, see Figure~\ref{fig:test-pairwise}(a).
	Second application resolved collision for UAS 1 and UAS 3 by UAS 3 deviating vertically downwards, see Figure~\ref{fig:test-pairwise}(b).
	The third application led to UAS 3 deviate vertically upwards, which resolved collision for UAS 2 and UAS 3, though created a re-appeared violation of minimum separation for UAS 1 and UAS 3 in the middle of their trajectories, see Figure~\ref{fig:test-pairwise}(c). 
\end{example}

\begin{figure}[b!]
	\begin{subfigure}[b]{0.3\columnwidth}
		\includegraphics[width=\textwidth]{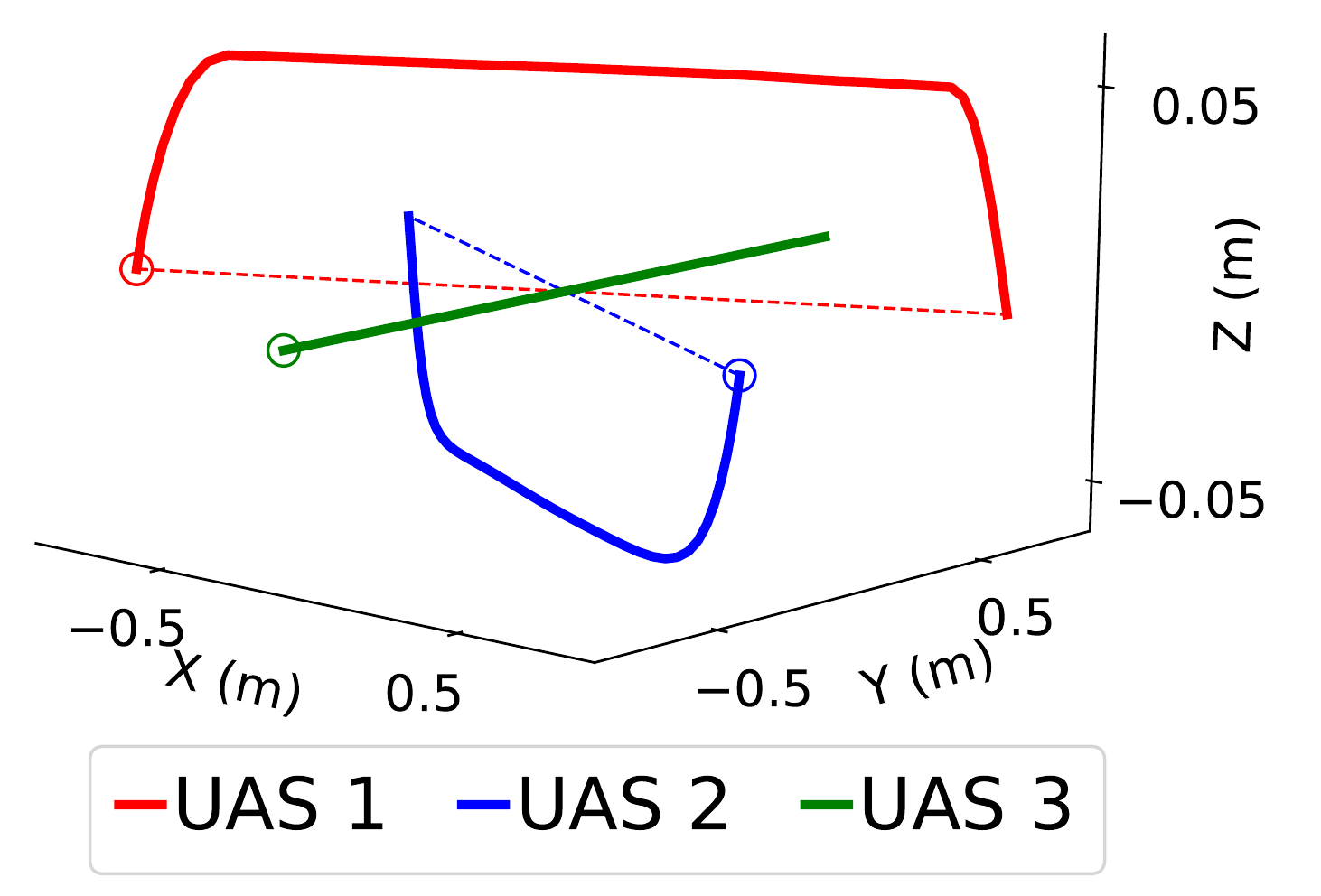}
	\caption{L2F for pair UAS 1, UAS 2. Pairwise separations: $\delta_{12}= 0.11$m, $\delta_{13}=0.06$m, $\delta_{23}= 0.05$m.}
	\end{subfigure}
	\hspace{5pt}
	\begin{subfigure}[b]{0.3\columnwidth}
		\includegraphics[width=\textwidth]{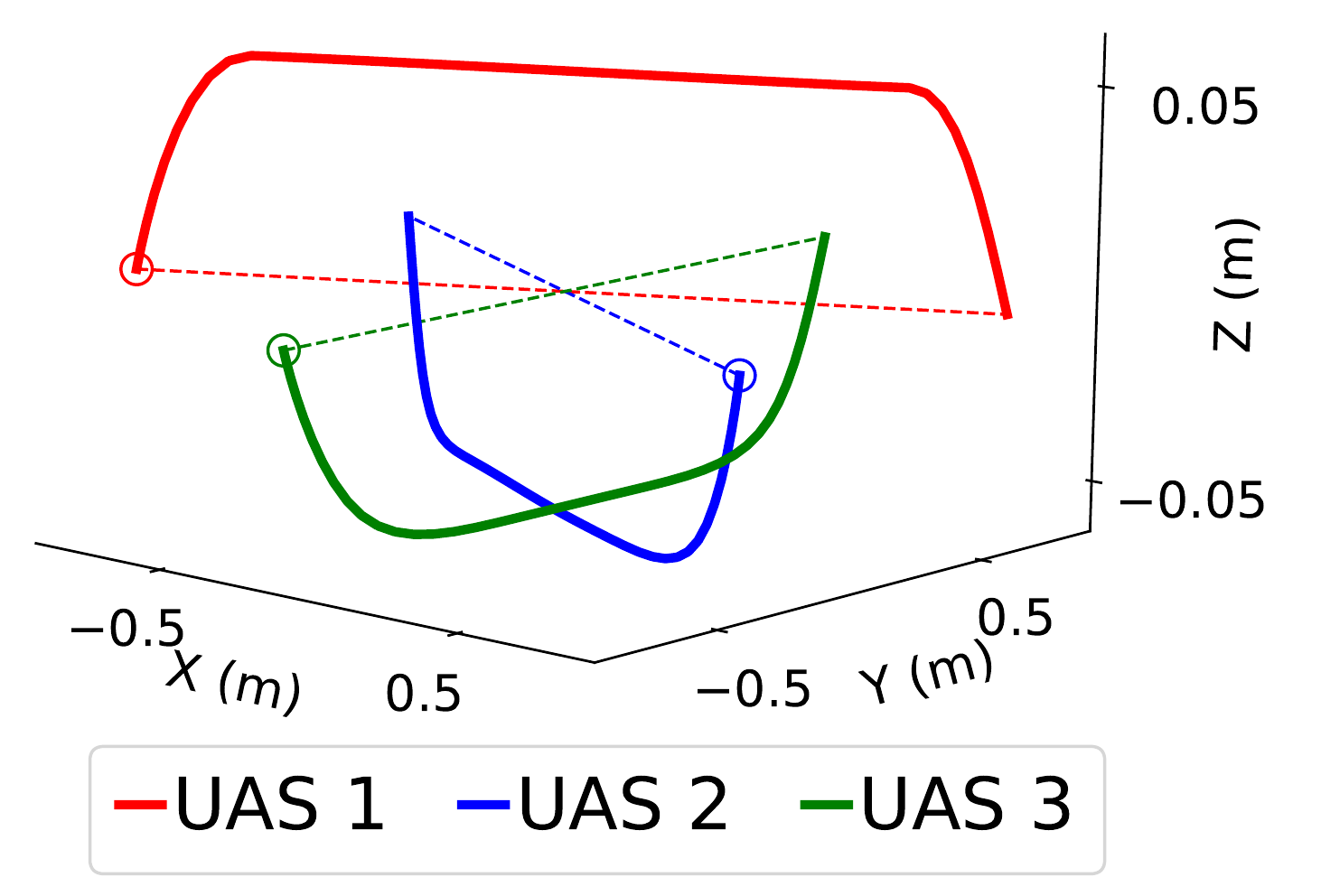}
		\caption{L2F for pair UAS 1, UAS 3. Pairwise separations: $\delta_{12}= 0.11$m, $\delta_{13}=0.11$m, $\delta_{23}= 0.04$m.}
	\end{subfigure}
	\hspace{5pt}
	\begin{subfigure}[b]{0.3\columnwidth}
		\includegraphics[width=\textwidth]{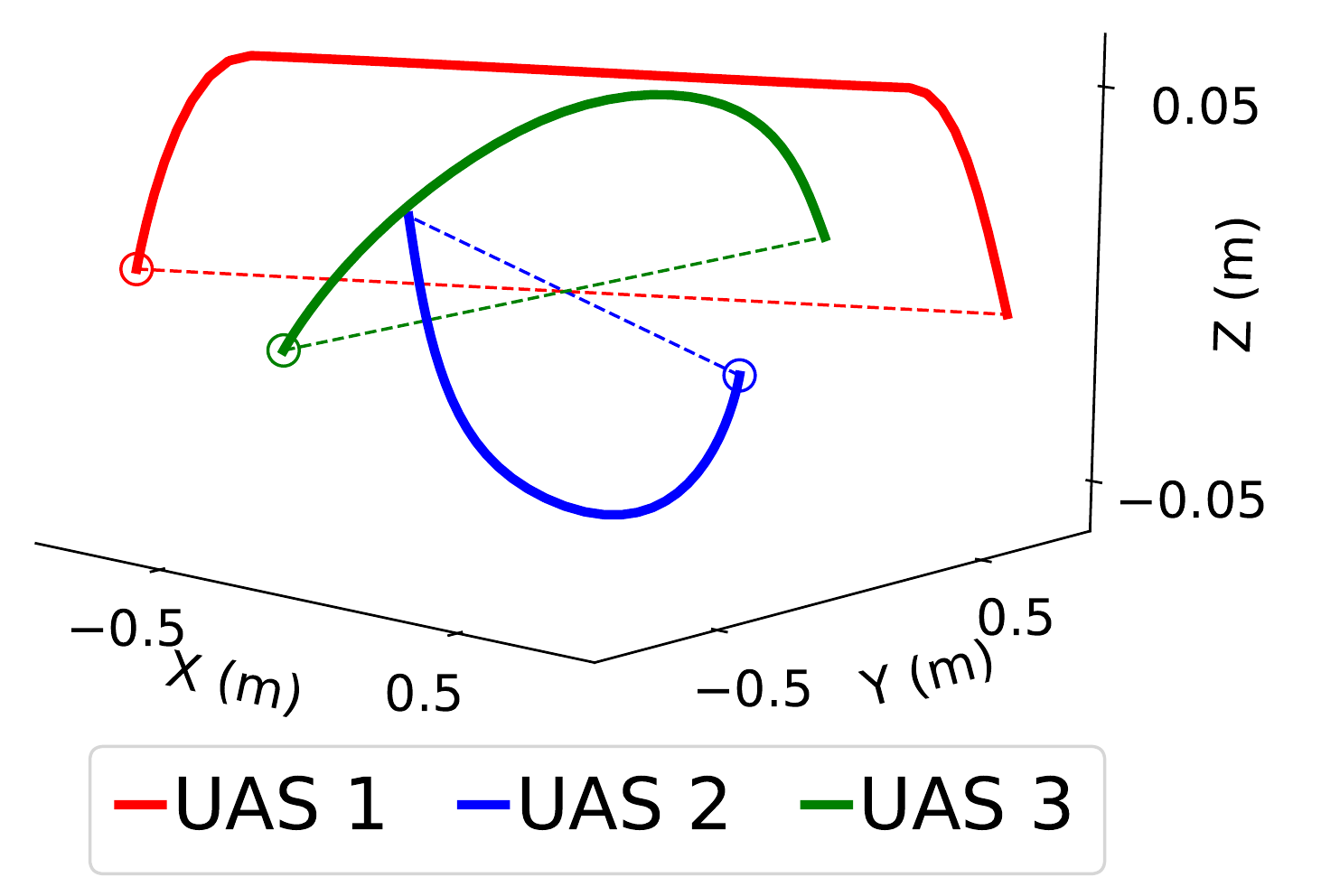}
		\caption{L2F for pair UAS 2, UAS 3 results. Pairwise separations: $\delta_{12}= 0.11$m, $\delta_{13}=0.01$m, $\delta_{23}= 0.1$m.}
	\end{subfigure}
	\vspace{-5pt}
	\caption{Sequential L2F application for the 3 UAS scenario. Pre-planned colliding trajectories are depicted in dashed lines. Simultaneous collision is detected at point $(0,0,0)$. The updated trajectories generated by L2F are depicted in solid color. Initial positions of UAS marked by ``O''. }
	\label{fig:test-pairwise}
\end{figure}

To overcome this \textit{live lock} like issue, where repeated pair-wise applications of L2F only result in new conflicts between other pairs of \ac{UAS}, we propose a modification of L2F called Learning-N-Flying (LNF). The LNF framework is based on pairwise application of L2F, but also incorporates a \emph{Robustness Tube Shrinking} (RTS) process described in Section~\ref{sec:shrinking} after every L2F application. The overall LNF framework is presented in Algorithm~\ref{alg:lnf}. 
Section~\ref{sec:experiments_lnf} presents extensive simulations to show the applicability of the LNF scheme to scenarios where more than two \ac{UAS} are on collisions paths, including in high-density \ac{UAS} operations.



\begin{algorithm}[t!]
	\SetAlgoLined
	\KwIn{Pre-planned fleet trajectories $\mathbf{x}_i$, initial robustness tubes $\rtube_i$, UAS priorities}
	\KwOut{New trajectories $\mathbf{x}_i'$, new robustness tubes $\rtube'_i$, control inputs $u'_{i,0}$}
	
	{Each UAS $i$ detects the set of UAS that it is in conflict with: $S=\{j\ |\ \exists k\ ||p_{i,k}-p_{j,k}|| < \delta, \ 0\leq k\leq H \}$}
	
	{Order $S$ by the UAS priorities}
	
	\For{$j\in S$}{
			{$(\mathbf{x}_i', \mathbf{x}_j', \mathbf{u}_i', \mathbf{u}_j')=\textbf{L2F}(\mathbf{x}_i, \mathbf{x}_j, \rtube_i, \rtube_j)$, see Section~\ref{sec:l2f}}
			
						
			\If{$\varUpsilon=\{ k:\ ||p_{i,k}'-p_{j,k}'|| < \delta, \ 0\leq k\leq H \}\not=\emptyset$}
			{
				{$(\mathbf{x}_i', \mathbf{x}_j', \mathbf{u}_i', \mathbf{u}_j') = \textbf{Repairing}(\mathbf{x}_i, \mathbf{x}_j, \rtube_i, \rtube_j, \varUpsilon)$}
			}
			
			{$(\rtube'_i, \rtube'_j)=\rts(\mathbf{x}'_i, \mathbf{x}'_j, \rtube_i,\rtube_j)$}		
		}
		{Apply controls $u_{i,0}'$ for the initial time step of the receding horizon}
	\caption{Learning-`N-Flying: Decentralized and cooperative collision avoidance for multi-UAS fleets. Applied in a receding horizon manner by each UAS $i$.}
	\label{alg:lnf}
\end{algorithm}

\subsection{Robustness tubes shrinking (RTS)}
\label{sec:shrinking}

\newcommand{\xone}{\mathbf{x}_1}
\newcommand{\xtwo}{\mathbf{x}_2}
\newcommand{\xthree}{\mathbf{x}_3}
\newcommand{\xpone}{\mathbf{x}_1'}
\newcommand{\xptwo}{\mathbf{x}_2'}
\newcommand{\xpthree}{\mathbf{x}_3'}

The high-level of idea of RTS is that, when two trajectories are de-collided by L2F, we want to constrain their further modifications by L2F so as not to induce new collisions. 
In Example~\ref{ex:rts}, after collision-free $\xpone$ and $\xptwo$ are produced by L2F and before $\xptwo$ and $\xthree$ are de-collided, we want to constrain any modification to $\xptwo$ s.t. it does not collide again with $\xpone$. 
Since trajectories are constrained to remain within robustness tubes, we simply shrink those tubes to achieve this.
The amount of shrinking is $\delta$, the minimum separation.
RTS is described in Algorithm~\ref{alg:shrinking}. 


\begin{algorithm}[t!]
	\SetAlgoLined
	\Notation{$(\rtube'_1, \rtube'_2)=\rts(\mathbf{x}'_1, \mathbf{x}'_2, \rtube_1,\rtube_2)$}
	\KwIn{New trajectories $\mathbf{x}'_1$, $\mathbf{x}'_2$ generated by L2F, initial robustness tubes $\rtube_1$, $\rtube_2$}
	\KwOut{New robustness tubes $\rtube'_1$, $\rtube'_2$}
	
	
	{Set $msep=\min_{0\leq k\leq H}||p_{1,k}'-p_{2,k}'||$}
	
	\For{$k=0,\ldots,H$}{
		\eIf{$\mathbf{dist}(P_{1, k}, P_{2, k})\geq \delta$}
		{
			{No shrinking required: $P'_{1,k}=P_{1,k},\ P'_{2,k}=P_{2,k}$}
		}
		{			
			
			{Determine the axis ($X$, $Y$ or $Z$) of maximum separation between $p'_{1,k}$ and $p'_{2,k}$}
			
			{Define the 3D box $\varPi_k$ with edges of size $\min(msep,\delta)$ along the determined axis and infinite edges along other two axes}
			
			{Center $\varPi_k$ at the midpoint between $p'_{1,k}$ and $p'_{2,k}$}
			
			{Remove $\varPi_k$ from both tubes:
				$P'_{1, k} = P_{1, k} \setminus \varPi_k,\ 
				P'_{2, k} = P_{2, k} \setminus \varPi_k$}
		}
	}
	\caption{Robustness tubes shrinking. Also see Figure~\ref{fig:tubes_shrinking}.}
	\label{alg:shrinking}
\end{algorithm}

\begin{figure}[t!]
	\centering
	\includegraphics[width=0.99\textwidth]{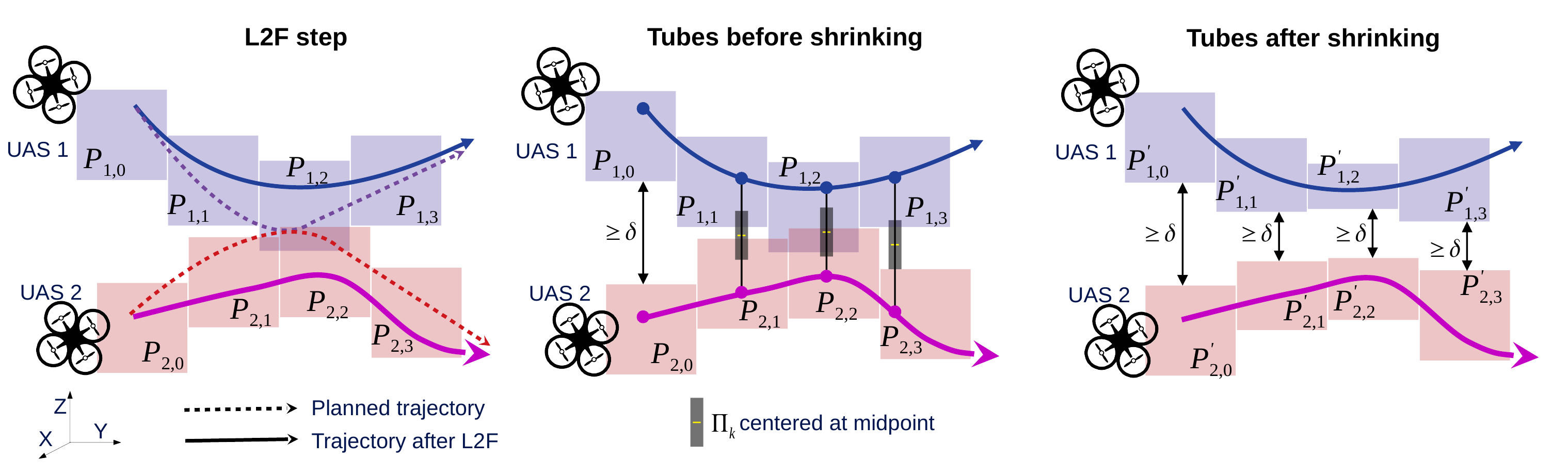}
	\vspace{-10pt}
	\caption{Visualization of the robustness tubes shrinking process.}
	\label{fig:tubes_shrinking}
\end{figure}

\begin{example}
	Figure~\ref{fig:tubes_shrinking}(a) presents the initial discrete-time robustness tubes and trajectories for UAS 1 and UAS 2. Successful application of L2F resolves the detected collision between initially planned trajectories $\mathbf{p}_1$, $\mathbf{p}_2$, depicted in dashed line. New non-colliding trajectories $\mathbf{p}_1'$ and $\mathbf{p}_2'$ produced by L2F are in solid color. Figure~\ref{fig:tubes_shrinking}(b) shows that for time step $k=0$ no shrinking is required since the robustness tubes $P_{1,0}$, $P_{2,0}$ are already $\delta$-separate. For time steps $k=1,2,3$, the axis of maximum separation between trajectories is $Z$, therefore, 
	boxes $\varPi_k$ are defined to be of height $\delta$ with infinite width and length. Boxes $\varPi_k$ are drawn in gray, midpoints between the trajectories are drawn in yellow. 
	Figure~\ref{fig:tubes_shrinking}(c) depicts the updated $\delta$-separate robustness tubes $\mathbf{P}'_1$ and $\mathbf{P}'_2$.
\end{example}


\begin{theorem}[Sufficient condition for $\delta$-separate tubes]
	\label{th:rts_success}
	Zero-slack solution of \eqref{eq:campc}
	implies that robustness tubes updated by RTS procedure are the subsets of the initial robustness tubes and $\delta$-separate, e.g. for robustness tubes $(\rtube'_1, \rtube'_2)=\rts(\mathbf{x}'_1, \mathbf{x}'_2, \rtube_1,\rtube_2)$, the following two properties hold:
	\begin{align}
	\label{eq:lnf_tube_dist}
	&\mathbf{dist}(\rtube_1', \rtube_2') \geq \delta\\
	\label{eq:lnf_tube_subset}
	& \rtube_j' \subseteq \rtube_j,\ \forall j\in \{1,2\} 
	\end{align}
\end{theorem}
See the proof in the appendix Section~\ref{sec:rts_appndx}.

\subsection{Combination of L2F with RTS}
Three following lemmas define important properties of L2F combined with the  shrinking process. Proofs can be found in the appendix Section~\ref{sec:rts_appndx}.
\begin{lemma}
	\label{lemma:p_in_tube}
	Let two trajectories $\mathbf{x}_1'$, $\mathbf{x}_2'$ be generated by L2F and let the robustness tubes $\rtube_1'$, $\rtube_2'$ be the updated tubes generated by RTS procedure from initial tubes $\rtube_1$, $\rtube_2$ using the trajectories $\mathbf{x}_1'$, $\mathbf{x}_2'$. Then 
	\begin{equation}
	\mathbf{p}_j' \in \rtube_j',\ \forall j\in \{1,2\}.
	\end{equation}
\end{lemma}
The above Lemma~\ref{lemma:p_in_tube} states that RTS procedure preserves trajectory belonging to the corresponding updated robustness tube.

\begin{lemma}
	\label{lemma:tubes}
	Let two robustness tubes $\rtube_1$ and $\rtube_2$ be $\delta$-separate. Then any pair of trajectories within these robustness tubes are non-conflicting, i.e.:
	\begin{equation}
	\forall \mathbf{p}_1\in \rtube_1,\ \forall\mathbf{p}_2\in \rtube_2,\ 
	||p_{1,k}-p_{2,k}|| \geq \delta, \, \forall k \in \{0,\dotsc,H\}.
	\end{equation}

\end{lemma}

Using Lemma~\ref{lemma:tubes} we can now prove that every successful application of L2F combined with the shrinking process results in new trajectories does not violate previously achieved minimum separations between UAS, unless the RTS process results in an empty robustness tube. In other words, it solves the 3 UAS issue raised in Example~\ref{ex:rts}. We formalize this result in the context of 3 UAS with the following Lemma:

\begin{lemma}
	\label{lemma:3uas}
	Let $\mathbf{x}_1, \mathbf{x}_2, \mathbf{x}_3$ be pre-planned conflicting UAS trajectories, and let $\rtube_1$, $\rtube_2$ and $\rtube_3$ be their corresponding robustness tubes. 
	Without loss of generality assume that the sequential pairwise application of L2F combined with RTS has been done in the following order:
	\begin{align}
		\label{eq:prop12}
		(\mathbf{x}_1', \mathbf{x}_2') = \ltf (\mathbf{x}_1, \mathbf{x}_2, \rtube_1, \rtube_2), &\qquad
		(\rtube_1', \rtube_2') = \rts (\mathbf{x}_1, \mathbf{x}_2, \rtube_1, \rtube_2) \\
		\label{eq:prop13}
		(\mathbf{x}_1'', \mathbf{x}_3') = \ltf (\mathbf{x}_1', \mathbf{x}_3, \rtube_1', \rtube_3), &\qquad
		(\rtube_1'', \rtube_3') = \rts (\mathbf{x}_1'', \mathbf{x}_3', \rtube_1', \rtube_3) \\
		\label{eq:prop23}
		(\mathbf{x}_2'', \mathbf{x}_3'') = \ltf (\mathbf{x}_2', \mathbf{x}_3', \rtube_2', \rtube_3'), &\qquad
		(\rtube_2'', \rtube_3'') = \rts (\mathbf{x}_2'', \mathbf{x}_3'', \rtube_2', \rtube_3')
	\end{align}
	If all three L2F applications gave zero-slack solutions then position trajectories $\mathbf{p}_1'', \mathbf{p}_2'', \mathbf{p}_3''$ pairwise satisfy mutual separation requirement, e.g.:
	\begin{align}
		\label{eq:ms12}
		||p''_{1,k}-p''_{2,k}||\geq \delta,\ \forall k\in\{0,\ldots,H\} \\
		\label{eq:ms13}
		||p''_{1,k}-p''_{3,k}||\geq \delta,\ \forall k\in\{0,\ldots,H\} \\
		\label{eq:ms23}
		||p''_{2,k}-p''_{3,k}||\geq \delta,\ \forall k\in\{0,\ldots,H\}
	\end{align}
and are within their corresponding robustness tubes:
	\begin{equation}
	\label{eq:im}
		\mathbf{p}''_{j}\in\rtube''_j,\ \forall j \in \{1,2,3\}. 
	\end{equation}
\end{lemma}

By induction we can extend Lemma~\ref{lemma:3uas} to any number of UAS. Therefore, we can conclude that for any $N$ pre-planned UAS trajectories, zero-slack solution of LNF is a sufficient condition for CA, e.g. resulting trajectories generated by LNF are non-conflicting and withing the robustness tubes of the initial trajectories. Note that this approach can still fail to find a solution, especially as repeated RTS can result in empty robustness tubes.

\begin{theorem}
\label{thm:LNF_termination}
For the case of $N$ UAS, when applied at any time step $k$, LNF (algorithm \ref{alg:lnf}) terminates after no more than $N \choose 2$ applications of pairwise L2F (algorithm \ref{alg:l2f}). 
\end{theorem}

This result follows directly from the inductive application of Lemma~\ref{lemma:3uas}. In experimental evaluations (Section \ref{sec:experiments_lnf}), we see that this worst-case number of L2F applications is not required often in practice. 

%

\section{Experimental evaluation of L2F and LNF}
\label{sec:exp}
In this section, we show the performance of our proposed methods via extensive simulations, as well as an implementation for actual quad-rotor robots. We compare L2F and L2F with repair (L2F+Rep) with the MILP formulation of Section \ref{sec:CA} and two other baseline approaches. Through multiple case studies, we show how LNF extends the L2F framework to work for scenarios with than two \ac{UAS}.

\subsection{Experimental setup}

\noindent \textbf{Computation platform:} All the simulations were performed on a computer with an AMD Ryzen 7 2700 8-core processor and 16GB RAM, running Python 3.6 on Ubuntu 18.04. 

\noindent \textbf{Generating training data:} We have generated the data set of 14K trajectories for training with collisions between UAS using the trajectory generator in \cite{mueller2015computationally}. The look-ahead horizon was set to $T=4$s and $dt=0.1$s. Thus, each trajectory consists of $H+1=41$ time-steps. The initial and final waypoints were sampled uniformly at random from two 3D cubes close to the fixed collision point, initial velocities were set to zero.

\noindent \textbf{Implementation details for the learning-based conflict resolution:}
The MILP to generate training data for the supervised learning of the CR scheme was implemented in MATLAB using Yalmip \cite{lofberg2004yalmip} with MOSEK v8 as the solver. The learning-based CR scheme was trained for $\rho = 0.055$ and minimum separation $\delta = 0.1$m which is close to the lower bound in Assumption \ref{assumption1}. This was implemented in Python 3.6 with Tensorflow 1.14 and Keras API and Casadi with qpOASES as the solver.  For traning the LSTM models (with different architectures) for CR, the number of training epochs was set to 2K with a batch size of 2K. Each network was trained to minimize categorical cross-entropy loss using Adam optimizer~\cite{kingma2014adam} with training rate of $\alpha=0.001$ and moment exponential decay rates of $\beta_1= 0.9$ and $\beta_2=0.999$.
The model with 3 LSTM layers with 128 neurons each, see Figure~\ref{fig:lstm_arc}, was chosen as the default learning-based CR model, and is used for the pairwise \ac{CA} approach of both L2F and LNF.

\noindent \textbf{Implementation details for the CA-MPC:}
For the online implementation of our scheme, we implement CA-MPC using CVXgen and report the computation times for this implementation. 
We then import CA-MPC in Python, interface it with the CR scheme and run all simulations in Python.


\subsection{Experimental evaluation of L2F}
\label{sec:experiments_l2f}

\begin{figure}[tb]
	\centering
	\begin{minipage}{.65\textwidth}
		\centering
		\includegraphics[width=\linewidth]{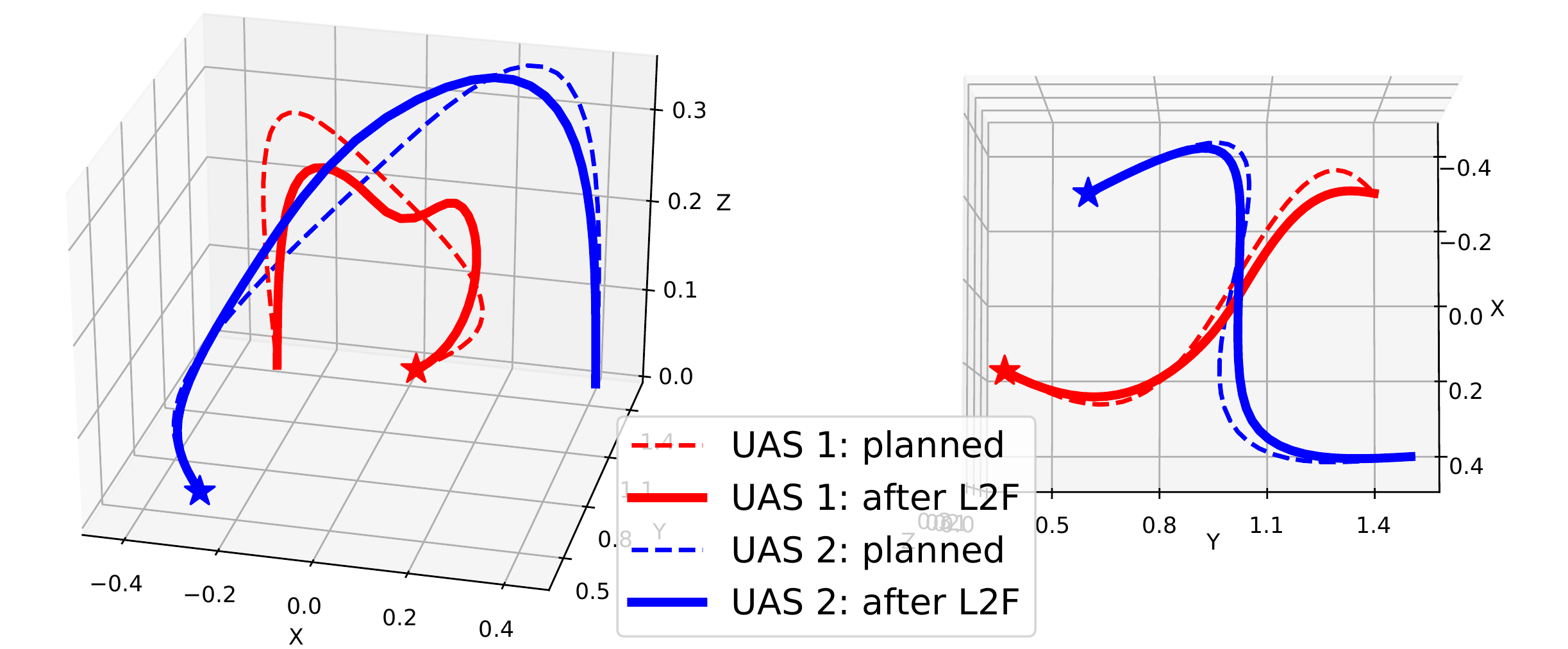}
		\captionof{figure}{Trajectories for 2 UAS from different angles. The dashed (planned) trajectories have a collision at the halfway point. The solid ones, generated through L2F method, avoid the collision while remaining within the robustness tube of the original trajectories. Initial UAS positions marked as stars. Playback of the scenario is at \url{https://tinyurl.com/l2f-exmpl}.}
		\label{fig:scen1_w_ca}
	\end{minipage}%
	\hfill
	\begin{minipage}{.32\textwidth}
		\centering
		\includegraphics[width=\linewidth]{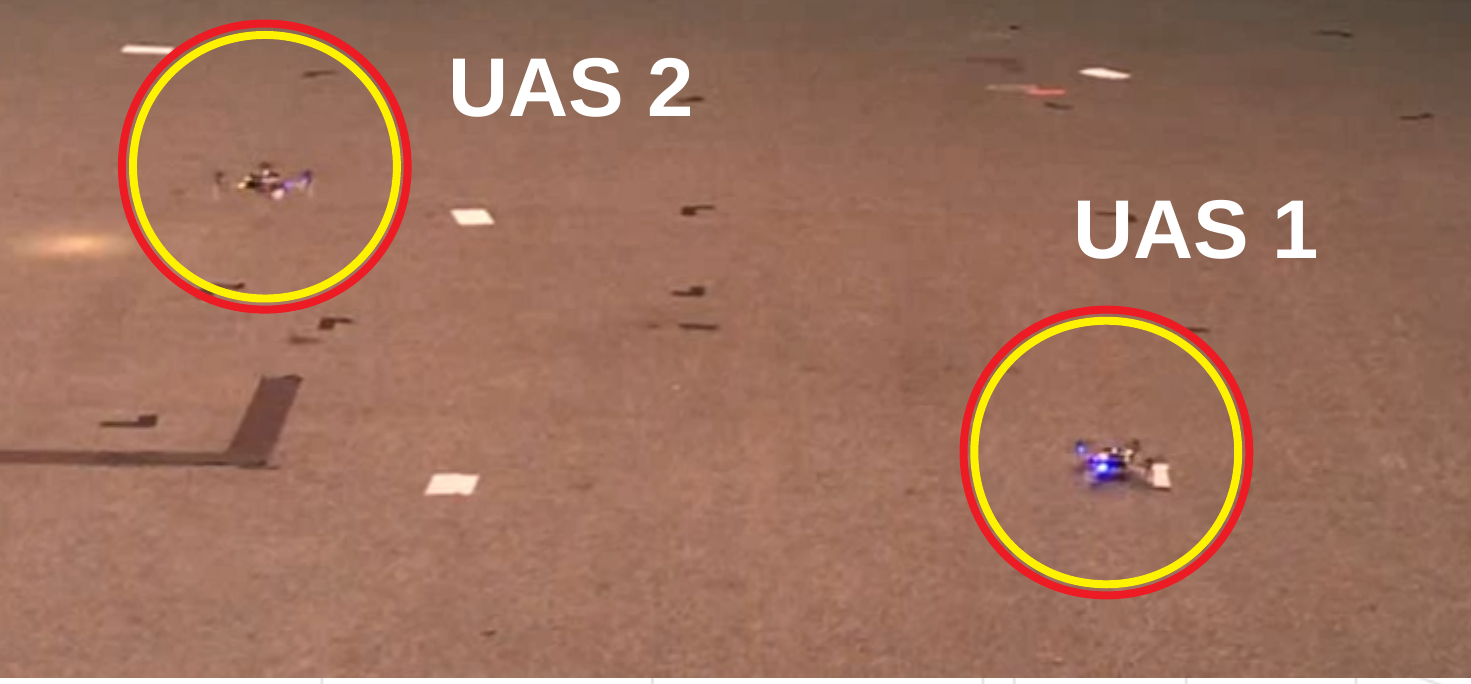}
		\includegraphics[width=\linewidth]{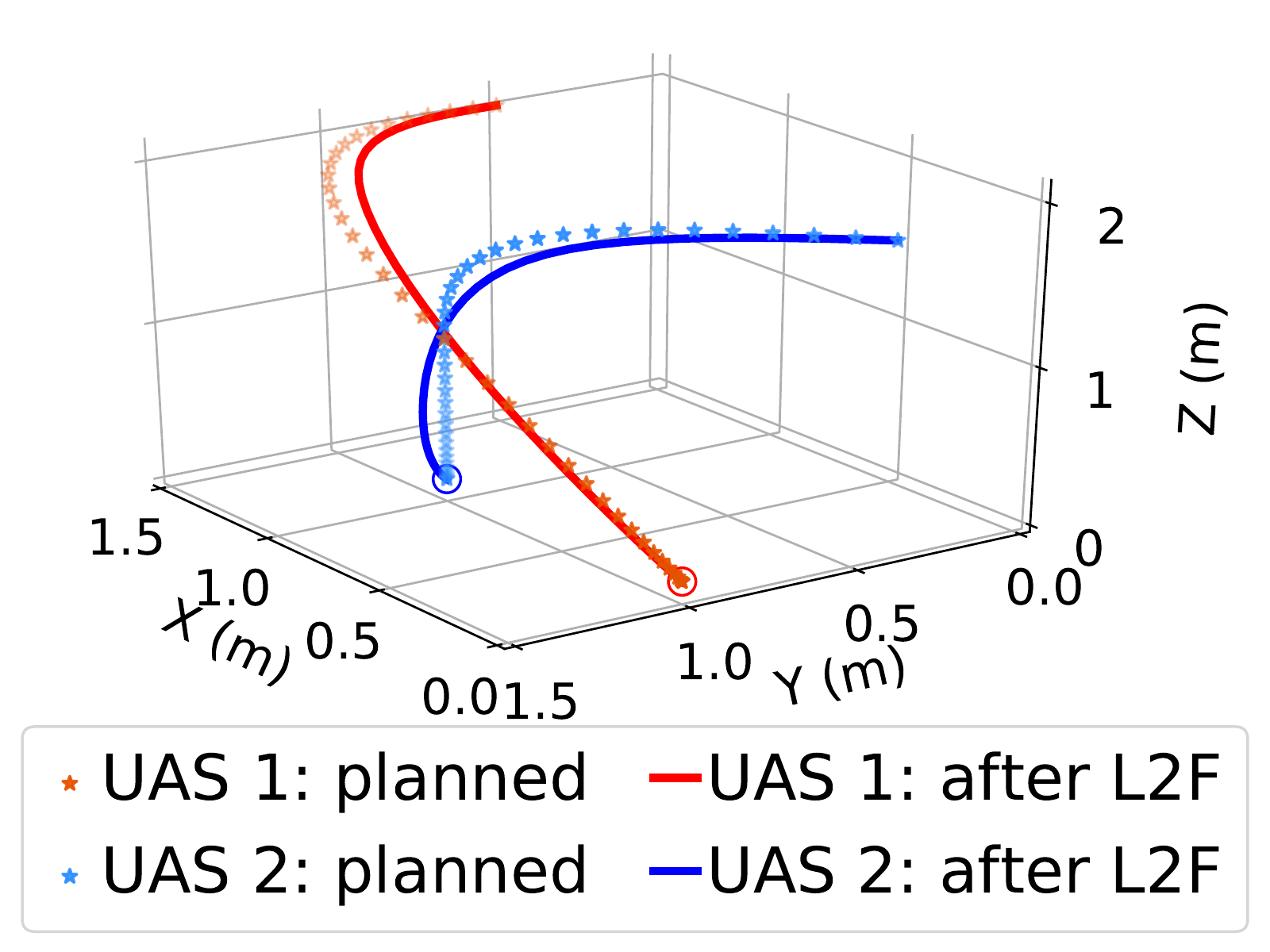}\vspace{-10pt}
		\captionof{figure}{Trajectories for 2 Crazyflie quad-rotors before (dotted) and after (solid) L2F. 
			Videos of this are at \url{https://tinyurl.com/exp-cf2}}
		\label{fig:crazyfly}
	\end{minipage}
\end{figure}

We evaluate the performance of L2F with 10K test trajectories (for pairwise \ac{CA}) generated using the same distribution of start and end positions as was used for training. Figure~\ref{fig:scen1_w_ca} shows an example of two UAS trajectories before and after L2F. Successful avoidance of the collision at the midway point on the trajectories can easily be seen on the playback of the scenario available at \url{https://tinyurl.com/l2f-exmpl}.
To demonstrate the feasibility of the deconflicted trajectories, we also ran experiments using two Crazyflie quad-rotor robots as shown in Figure~\ref{fig:crazyfly}.  
Videos of the actual flights and additional simulations can be found at \url{https://tinyurl.com/exp-cf2}.

\subsubsection{Results and comparison to other methods}
\label{sec:exp_results}

We analyzed three other methods alongside the proposed learning-based approach for L2F.
\begin{enumerate}
	\item A \textbf{random} decision approach which outputs a 
	sequence sampled from the discrete uniform distribution.
	\item A \textbf{greedy} approach that selects the discrete decisions that correspond to the direction of the most separation between the two UAS at each time step. 
	For more details see \cite{itsc20}.
	
	\item A \textbf{L2F with Repairing} approach following Section~\ref{sec:repair}.
	\item A centralized \textbf{MILP} solution that picks decisions corresponding to binary decision variables in \eqref{eq:CentralMILP}.
\end{enumerate}


\begin{figure}[t !]
	\begin{subfigure}[b]{0.45\columnwidth}
		\includegraphics[width=\textwidth]{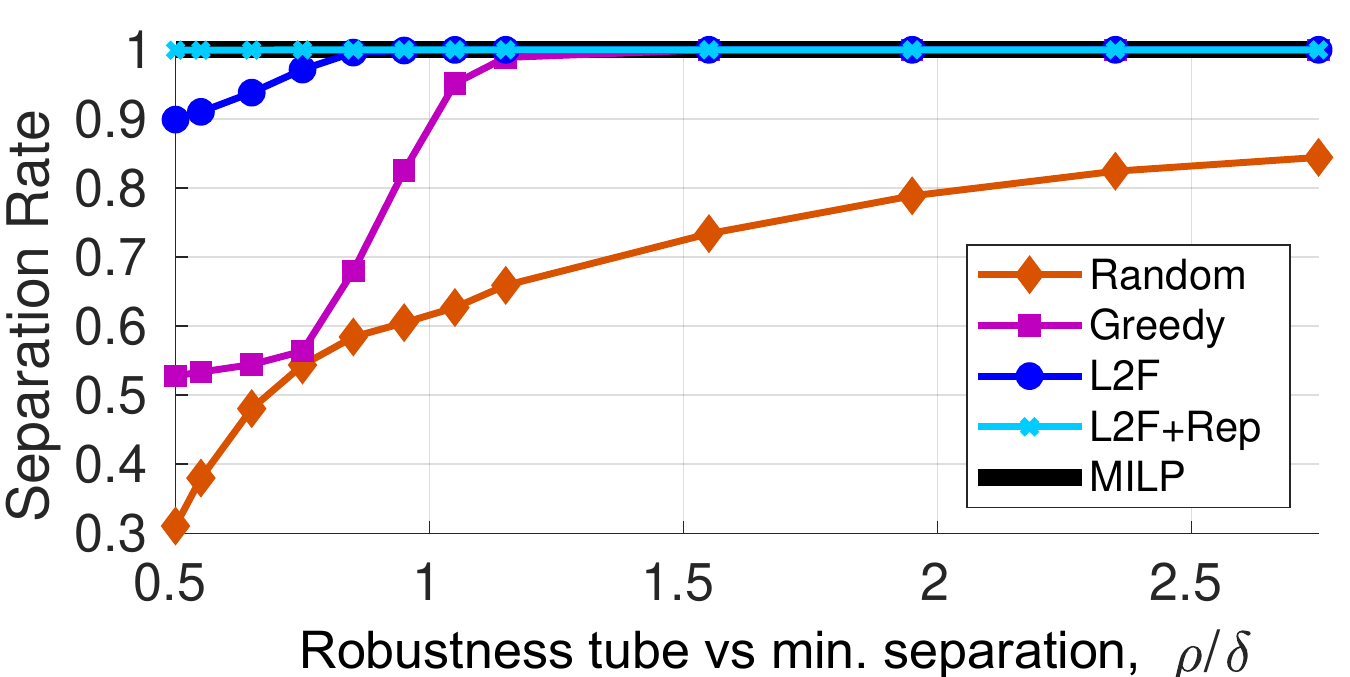}
		\caption{\textit{Separation rate} defines the fraction of  initially conflicting trajectories for which UAS managed to achieve minimum separation.}
	\end{subfigure}
	\hspace{10pt}
	\begin{subfigure}[b]{0.45\columnwidth}
		\includegraphics[width=\textwidth]{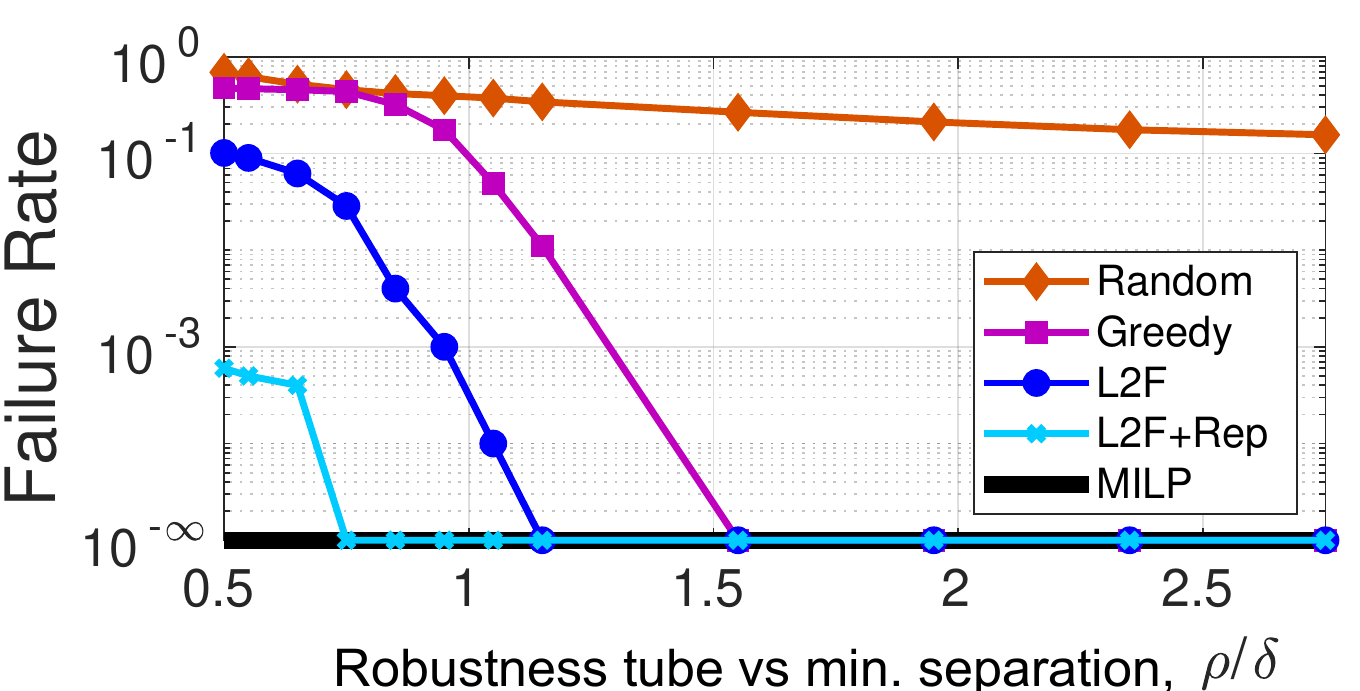}
		\caption{\textit{Failure rate} (1-\textit{Separation rate}) defines the fraction of initially conflicting trajectories for which UA could not achieve minimum separation.}
	\end{subfigure}
	\vspace{-5pt}
	\caption{\small Model sensitivity analysis with respect to variations of fraction $\rho/\delta$ which connects the minimum allowable robustness tube radius $\rho$ to the minimum allowable separation between two UAS $\delta$, see Assumption~\ref{assumption1}. A higher $\rho/\delta$ implies there is more room within the robustness tubes to maneuver for CA.}
	\label{fig:rho_rate_33}
\end{figure}

For the evaluation, we measured and compared the \textbf{separation rate} and the \textbf{computation time} for all the methods over the same 10K test trajectories. \textit{Separation rate} defines the fraction of the conflicting trajectories for which UAS managed to achieve minimum separation after a \ac{CA} approach. Figure~\ref{fig:rho_rate_33} shows the impact of the $\rho/\delta$ ratio on separation rate. Higher $\rho/\delta$ implies wider robustness tubes for the UAS to maneuver within, which should make the CA task easier as is seen in the figure. %
%
The centralized \ac{MILP} has a separation rate of $1$ for each case here, however is unsuitable for an online implementation with its computation time being over a minute ( seetable \ref{tbl:sep-rate-time}) and we exclude it from the comparisons in the text that follows. In the case of $\rho/\delta=0.5$, where the robustness tubes are just wide enough to fit two UAS (see Assumption~\ref{assumption1}), we see the L2F with repairing (L2F+Rep) significantly outperforms the methods. This worst-case performance of L2F with repairing is $0.999$ which is significantly better than the other approaches including the original L2F. As the ratio grows, the performance of all methods improve, with L2F+Rep still outperforming the others and quickly reaching a separation rate of $1$. For $\rho/\delta \geq 1.15$, L2F no longer requires any repair and also has a separation rate of $1$. 
 
Table \ref{tbl:sep-rate-time} shows the separation rates for three different $\rho/\delta$ value as well as the computation times (mean and standard deviation) for each CA algorithm. L2F and L2F+Rep have an average computation time of less than $10$ms, making them suited for an online implementation even at our chosen control sampling rate of $10$Hz. For all CA schemes excluding MILP, the smaller the $\rho/\delta$ ratio, the more UAS 1 alone is unsuccessful at collision avoidance MPC~\eqref{eq:drone1mpc}, and UAS 2 must also solve its CA-MPC~\eqref{eq:drone2mpc} and deviate from its pre-planned trajectory.
Therefore, computation time is higher for smaller $\rho/\delta$ ratio and lower for higher $\rho/\delta$ values. A similar trend is observed for the MILP, even though it jointly solves for both UAS, showing that it is indeed harder to find a solution when the $\rho/\delta$ ratio is small.

\begin{table}[]
	\renewcommand{\arraystretch}{1.3}
	\setlength{\tabcolsep}{1.5pt}
	\centering
	\begin{tabular}{l|c||c|c|c|c|c|}
		\cline{3-7}
		\multicolumn{1}{c}{\textbf{}} &
		\multirow{2}{*}{\textbf{}} &
		\multicolumn{5}{c|}{\textbf{CA Scheme}} \\ \cline{3-7} 
		\multicolumn{1}{c}{\textbf{}} &
		&
		\textbf{Random} &
		\textbf{Greedy} &
		\textbf{L2F} &
		\textbf{L2F+Rep} &
		\textbf{MILP} \\ \hline\hline
		\multicolumn{1}{|l|}{\multirow{3}{*}{\textbf{Separation rate}}} 
		&
		$\boldsymbol{\rho}/\boldsymbol{\delta}=\textbf{0.5}$ 
		&0.311
		&0.528
		&0.899
		&0.999
		&1
		\\ \cline{2-7} 
		\multicolumn{1}{|l|}{} 
		& $\boldsymbol{\rho}/\boldsymbol{\delta}=\textbf{0.95}$ 
		&0.605 
		&0.825
		&0.999
		& 1 
		& 1 \\ \cline{2-7} 
		\multicolumn{1}{|l|}{} &
		$\boldsymbol{\rho}/\boldsymbol{\delta}=\textbf{1.15}$ 
		& 0.659 & 0.989 & 1 & 1 & 1 \\ \hline \hline
		\multicolumn{1}{|l|}{\multirow{3}{*}{
				\begin{tabular}[c]{@{}l@{}}\textbf{Comput. time}  (ms) \\ (mean $\pm$ std)\end{tabular}
			}} 
		&
		$\boldsymbol{\rho}/\boldsymbol{\delta}=\textbf{0.5}$ 
		&$7.9\pm 0.01$ 
		&$9.7\pm0.6$
		&$9.1\pm1.3$
		&$9.7\pm3.6 $
		&$(98.9\pm 44.9)\cdot 10^3$
		\\ \cline{2-7} 
		\multicolumn{1}{|l|}{} 
		& $\boldsymbol{\rho}/\boldsymbol{\delta}=\textbf{0.95}$ 
		&$ 7.5\pm0.01$
		&$ 9.3\pm0.5$
		&$8.7\pm0.5$
		&$8.7\pm0.5$
		&$(82.5\pm36.3)\cdot 10^3$  \\ \cline{2-7} 
		\multicolumn{1}{|l|}{} &
		$\boldsymbol{\rho}/\boldsymbol{\delta}=\textbf{1.15}$ 
		& $6.3\pm1.9$
		& $7.1\pm2.$
		&$8.6\pm0.5$
		& $8.7\pm0.4$ 
		& $(33.1\pm34.9)\cdot 10^3$ \\ \hline
	\end{tabular}
	\caption{Separation rates and computation times (mean and standard deviation) comparison of different CA schemes. \textit{Separation rate} is the fraction of conflicting trajectories for which separation requirement
	\eqref{eq:msep} is satisfied after CA. \textit{Computation time} estimates the overall time demanded by CA scheme. MILP reports the time spent on solving~\eqref{eq:CentralMILP}. Other CA schemes report time needed for CR and CA-MPC together. L2F with repairing includes repairing time as well.}
	\label{tbl:sep-rate-time}
\end{table}

\subsection{Experimental evaluation of LNF}
\label{sec:experiments_lnf}
Next, we carry out simulations to evaluate the performance of LNF, especially in terms of scalability to cases with more than two \ac{UAS} and analyze its performance in wide variety of settings.
\subsubsection{Case study 1: Four UAS position swap}
\label{sec:case_pos_swap}
We recreate the following experiment from \cite{alonso2015collision}.
Here, two pairs of UAS must maneuver to swap their positions, i.e. the end point of each UAS is the same as the starting position for another UAS. See the 3D representation of the scenario in Figure~\ref{fig:pos_swap_3d}(a). Each UAS start set is assumed to be a singular point fixed at:
\begin{equation}
\textit{Goal}_1 = (1, 0, 0),\ \textit{Goal}_2= (0, 1, 0),\ \textit{Goal}_3=(-1, 0, 0), \ \textit{Goal}_4=(0, -1, 0)
\end{equation}
and goal states are antipodal to the start states:
\begin{equation}
\textit{Start}_j=-\textit{Goal}_j,\ \forall j\in\{1,2,3,4\}.
\end{equation}

All four UAS must reach desired goal states within 4 seconds while avoiding each other. With a pairwise separations requirement of at least $\delta=0.1$ meters, the overall mission specification is:
\begin{equation}
\label{eq:pos_mission}
\varphi_{\textit{mission}} = \bigwedge_{j=1}^4 \eventually_{[0,4]} (\mathbf{p}_j \in \textit{Goal}_j)\ \wedge\ 
\bigwedge_{j\not=j'} \always_{[0,4]}||\mathbf{p}_j-\mathbf{p}_{j'}||\geq 0.1
\vspace{-3pt}
\end{equation}


Following Section~\ref{sec:problem_planning}, initial pre-planning is done by ignoring the mutual separation requirement in \eqref{eq:pos_mission} and
generating the trajectory for each UAS $j=\{1,2,3,4\}$ independently with respect to its individual STL specification:
\begin{equation}
\varphi_j =  
\eventually_{[0,4]} (\mathbf{p}_j \in \textit{Goal}_j).
\end{equation}
Obtained pre-planned trajectories 
contain a joint collision that happens simultaneously (at $t=2$s, see Figure~\ref{fig:pos_swap_dist}) across all four UAS and located at point $(0,0,0)$, see Figure~\ref{fig:pos_swap_3d}(b).
For LNF experimental evaluation, the robustness value was fixed at $\rho=0.055$ and the UAS priorities were set in the increasing order, e.g. UAS with a lower index has the lower priority: $1<2<3<4$.

\begin{figure}[t!]
	\centering
	\includegraphics[width=.9\textwidth]{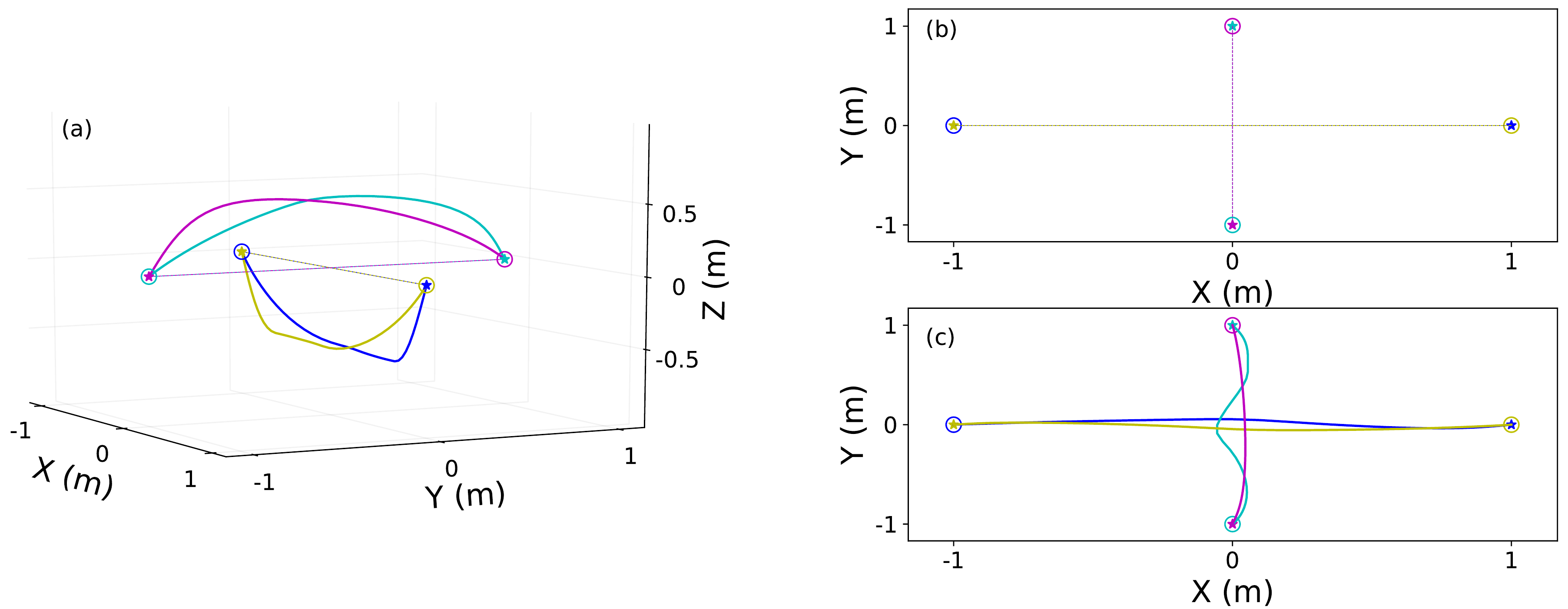}
	\caption{Four UAS position swap. (a): 3D representation of the scenario. (b)-(c): 2D projections of the scenario onto the horizontal plane $XoY$ before and after collision avoidance. Initial colliding trajectories are depicted in dashed lines in (a) and (b). Collision is detected at point $(0,0,0)$, it involves all four UAS and happens simultaneously across the agents. The updated non-colliding trajectories generated by LNF are depicted in solid color in (a) and (c). Initial positions of UAS marked by ``O'' and final positions by ``$\star$''.}
	\label{fig:pos_swap_3d}
\end{figure}


\begin{figure}[t!]
	\includegraphics[width=.95\textwidth]{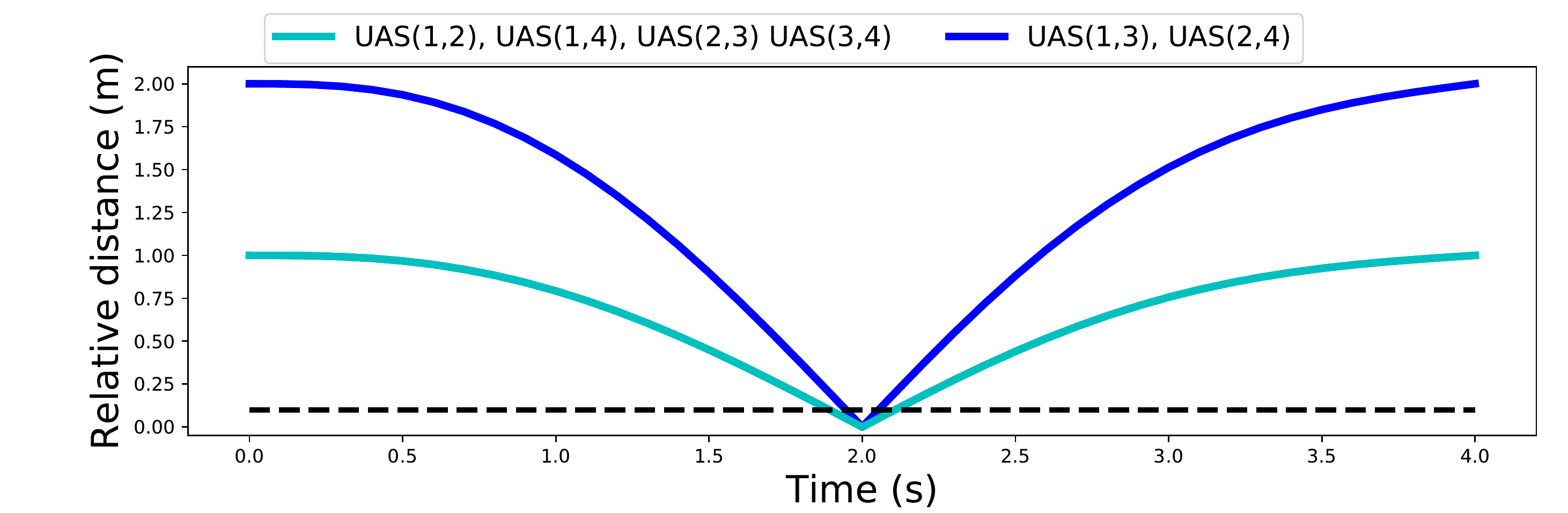}\\
	\includegraphics[width=.95\textwidth]{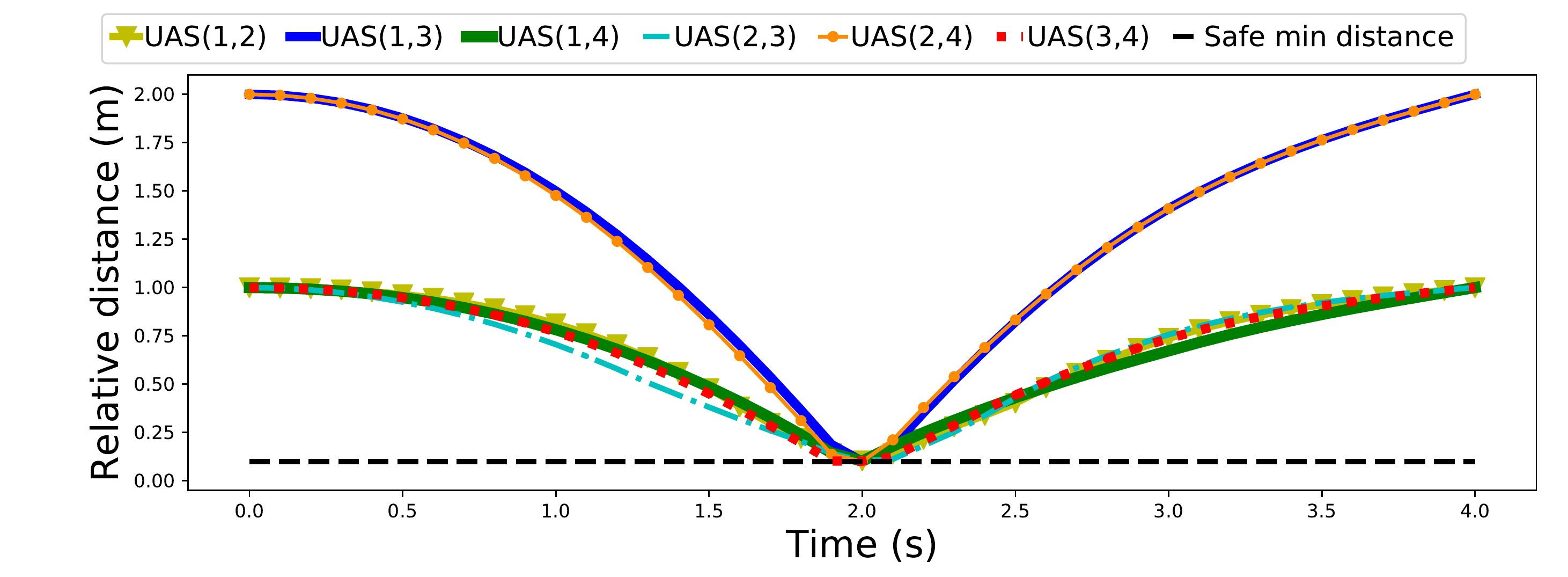}
	\caption{\small Four UAS position swap: Relative distances before (top) and after (bottom) the collision avoidance algorithm. Initial simultaneous collisions across all four UAS are successfully resolved by LNF. Note that the symmetry in the initial positions and trajectories results in multiple UAS pairs with the same relative distances for the time horizon of interest before collision avoidance (top).}
	\label{fig:pos_swap_dist}
	\vspace{15pt}
\end{figure}

\textbf{Simulation results.}
The non-colliding trajectories generated by LNF are depicted in
Figure~\ref{fig:pos_swap_3d}(c).
Playback of the scenario can be found at 
\url{https://tinyurl.com/swap-pos}.

It is observed that the opposite UAS pairs chose to change attitude and pass over each other, see Figure~\ref{fig:pos_swap_3d}(a). Within these opposite pairs, UAS chose to have horizontal deviations to avoid collision, see Figure~\ref{fig:pos_swap_3d}(c).
LNF algorithm performed $4 \choose 2$$=6$ pairwise applications of L2F (see Theorem~\ref{thm:LNF_termination}). 
Such high number of applications is expected due to a complicated simultaneous nature of the detected collision across the initially pre-planned trajectories. 
No CR repairing was required to successfully produce non-colliding trajectories by the LNF algorithm.
It took LNF $37.8$ms to perform CA. 
Figure~\ref{fig:pos_swap_dist} represents relative distances between UAS pairs before and after collision avoidance. Figure~\ref{fig:pos_swap_dist} shows that none of the UAS cross the safe minimum separation threshold of $0.1$m after LNF, e.g. joint collision has been successfully resolved by LNF.





\subsubsection{Case study 2: Four UAS reach-avoid mission}
\label{sec:case_reach_avoid}

Figure~\ref{fig:scenario_wall} depicts a multi UAS case-study with a 
reach-avoid mission. Scenario consists of four UAS which must 
reach desired goal states within 4 seconds while avoiding the wall 
obstacle and each other. Each UAS $j\in\{1,\ldots,4\}$ specification can be defined as:
\begin{equation}
\label{eq:scenario_spec_d}
\varphi_j =  
\eventually_{[0,4]} (\mathbf{p}_j \in \textit{Goal}_j)
\ \wedge\ 
\always_{[0,4]} \neg (\mathbf{p}_j \in \textit{Wall})
\vspace{-2pt}
\end{equation}
A pairwise separations requirement of $\delta=0.1$ meters is enforced for 
all UAS, therefore, the overall mission specification is:
\vspace{-2pt}
\begin{equation}
\label{eq:case_mission}
\varphi_{\text{mission}} = \bigwedge_{j=1}^4 \varphi_j\ \wedge\ 
\bigwedge_{j\not=j'} \always_{[0,4]}||\mathbf{p}_j-\mathbf{p}_{j'}||\geq 0.1
\vspace{-3pt}
\end{equation}


\begin{figure}[tb]
	\begin{subfigure}[b]{0.6\columnwidth}
		\includegraphics[width=\textwidth]{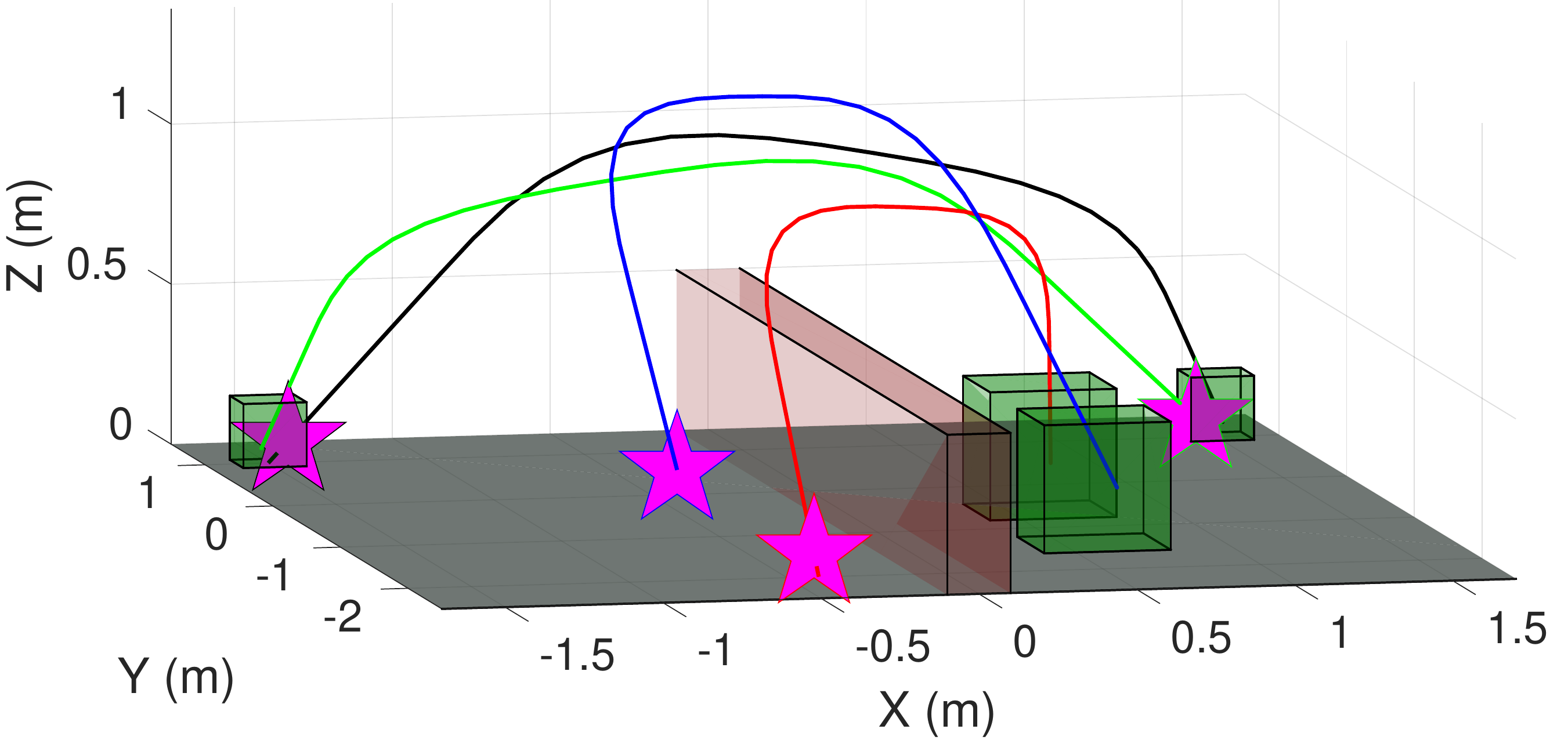}
		\caption{3D representation of the scenario}
		\label{fig:w11}
	\end{subfigure}
	\begin{subfigure}[b]{0.29\columnwidth}
		\includegraphics[width=\textwidth]{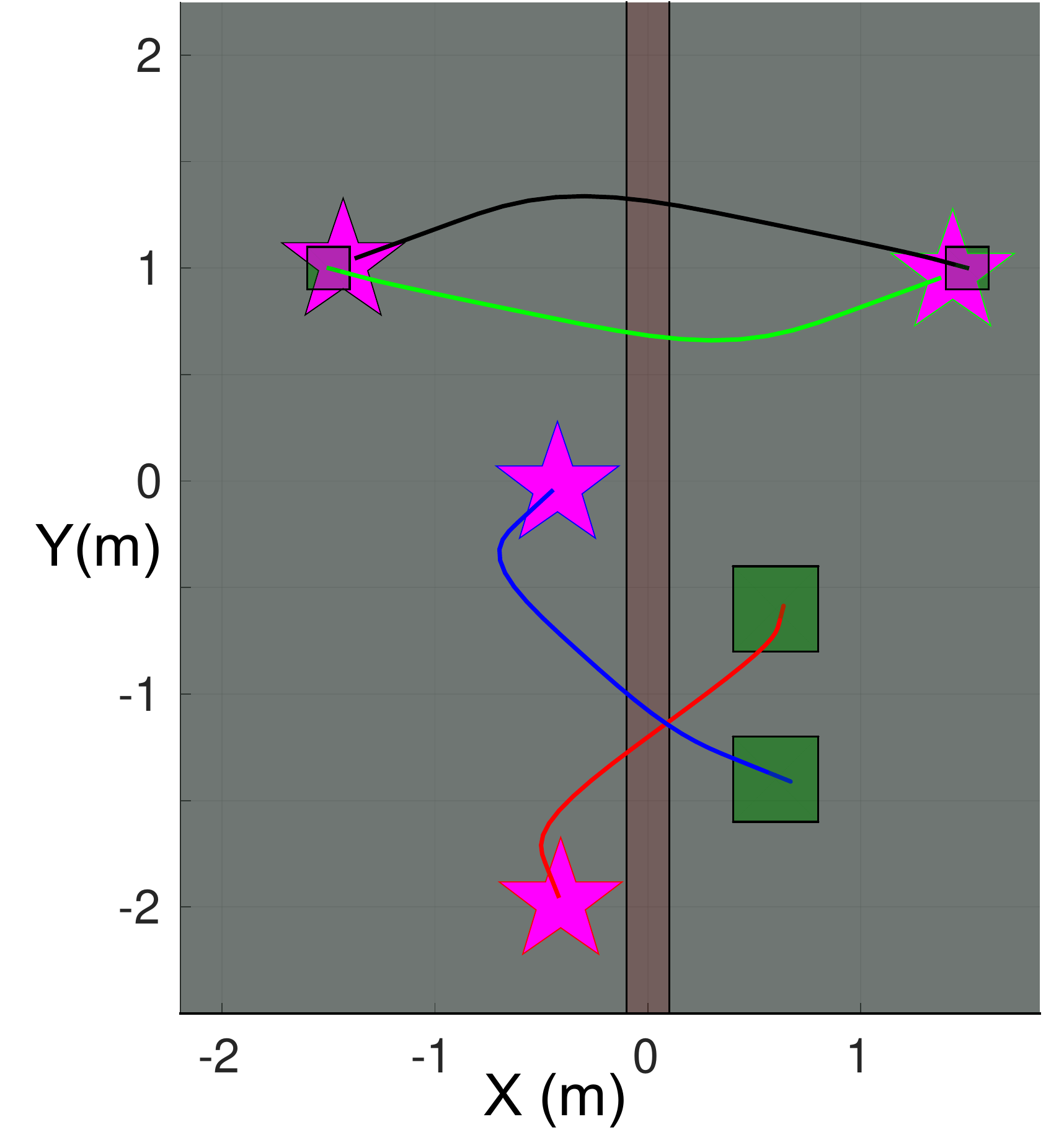}
		\caption{2D projection onto $XoY$}
		\label{fig:w22}
	\end{subfigure}
	\vspace{-10pt}
	\caption{Reach-avoid mission.
		Non-colliding trajectories for 4 UAS generated by LNF. All UAS reach their goal sets 
		(green boxes) within 4 seconds, do not crash into the vertical wall (in red) and satisfy pairwise separation requirement 
		of $0.1$m. Initial UAS positions marked by magenta ``$\star$''. Simulations are available at 
		\url{https://tinyurl.com/reach-av}.}
	\label{fig:scenario_wall}
	\vspace{-10pt}
\end{figure}

%

First, we solved the planning problem for all four UAS in a centralized manner following 
approach from~\cite{pant2018fly}. 
Next, we solved the planning problem for each UAS $j$ and its specification $\varphi_j$ independently, with calling LNF on-the-fly, after planning is complete. This way, independent planning with the online collision avoidance scheme guarantees the satisfaction of the overall mission specification \eqref{eq:case_mission}.

\textbf{Simulation results.}
We have simulated the scenario for 100 different initial conditions.
Computation time results are presented in Table~\ref{tbl:wall}.
The average computation time to generate trajectories in a centralized manner was $0.35$ seconds. The average time per UAS when planning independently (and in parallel) was $0.1$ seconds. 
These results demonstrate a speed up of $3.5\times$ for the individual UAS planning versus centralized \cite{pant2018fly}.
Scenario simulations are available  \url{https://tinyurl.com/reach-av}.

\begin{table}[t]
	\vspace{-10pt}
	\renewcommand{\arraystretch}{1.3}
	\setlength{\tabcolsep}{2.5pt}
	\centering
	\begin{tabular}{l|c|c|c}
		\multicolumn{1}{c|}{\textbf{}} &
		\multirow{2}{*}{Centralized planning \cite{pant2018fly}} &
		\multicolumn{2}{c}{Decentralized planning with CA} \\ \cline{3-4} 
		&
		&
		\multicolumn{1}{l|}{Independent planning} &
		{CA with LNF} \\ \hline
		Comput. time (mean$\pm$ std)(ms)
		&345.8$\pm$ 87.2
		& 138.6$\pm$ 62.4 & 
		9.97 $\pm$ 0.4
	\end{tabular}
	\caption{Reach-avoid mission. Computation times (mean and standard deviation) comparison between centralized planning following~\cite{pant2018fly} and decentralized planning (independent planning with LNF) over 100 runs of the scenario.}
	\label{tbl:wall}
\end{table}





%

\subsubsection{Case study 3: UAS operations in high-density airspace}
\label{sec:case_unit_cube}


\begin{figure}
	\vspace{-15pt}
	\begin{minipage}[c]{0.45\textwidth}
		\includegraphics[width=\textwidth]{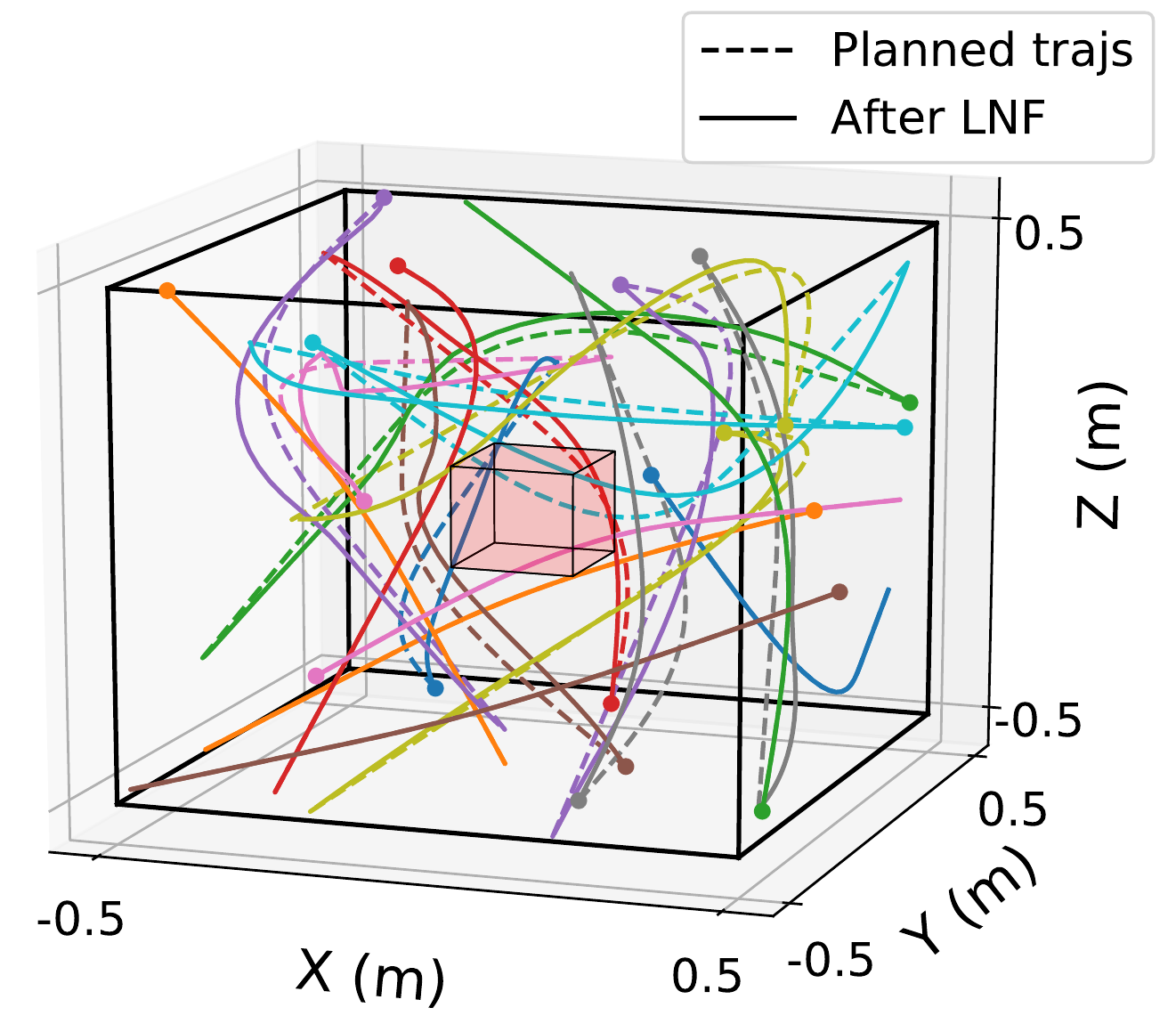}
	\end{minipage}\hfill
	\begin{minipage}[c]{0.55\textwidth}
		\caption{3D representation of the unit cube scenario with 20 UAS. All UAS must reach their goal sets within 4 seconds, avoid the no-fly zone and satisfy pairwise separation requirement of $0.1$m. Initially planned trajectories (dashed lines) had 5 violations of the mutual separation requirement. LNF succesfully resolved all detected violations and led to non-colliding trajectories (solid lines). Simulations
			are available at \url{https://tinyurl.com/unit-cube}.
		} \label{fig:cube_3D}
	\end{minipage}
\vspace{-15pt}
\end{figure}


To verify scalability of LNF, we perform evaluation of the scenario with high-density UAS operations.  
The case study consists of multiple UAS flying within the restricted area of 1m$^3$ while avoiding a no-fly zone of $(0.2)^3$=0.08m$^3$ in the center, see Figure~\ref{fig:cube_3D}.
Such scenario represents a hypothetical constrained and highly populated airspace with heterogeneous UAS missions such as package delivery or aerial surveillance.

Each UAS' $j$ start position $\textit{Start}_j$ and goal set $\textit{Goal}_j$ are chosen at (uniform) random on the opposite random faces of the unit cube.
Goal state should be reached within $4$ second time interval and the no-fly zone must be avoided during this time interval.
Same as in the previous case studies, we first solve the planning problem for each UAS $j$ separately following trajectory generation approach from~\cite{pant2018fly}.
The STL specification for UAS $j$ is captured as follows:
\begin{equation}
\label{eq:cube_spec_d}
\varphi_j =  
\eventually_{[0,4]} (\mathbf{p}_j \in \textit{Goal}_j)
\ \wedge\ 
\always_{[0,4]} \neg (\mathbf{p}_j \in \textit{NoFly})
\end{equation}
After planning is complete and trajectories $\mathbf{p}_j$ are generated, we call LNF on-the-fly to satisfy the overall mission specification $\varphi_{\text{mission}} = \bigwedge_{j=1}^N \varphi_j\ \wedge\ 
\varphi_{\text{separation}}$,
where $N$ is a number of UAS participating in the scenario and 
$\varphi_{\text{mission}}$ is the requirement of the minimum allowed pairwise separation of $0.1$m between the UAS:
\begin{equation}
\varphi_{\text{separation}}= \bigwedge_{j,j':\ j\not=j'} \always_{[0,4]}||\mathbf{p}_j-\mathbf{p}_{j'}||\geq 0.1.
\end{equation}


We increase the density of the scenario by increasing the number of UAS, while keeping the space volume at 1m$^3$.

\textbf{Simulation results.}
We ran 100 simulations for various numbers of UAS, each with randomized start and goal positions. Trajectory pre-planning was done independently for all UAS, and LNF is tasked with \ac{CA}. For evaluation, we measure the overall number of minimum separation requirement violations before and after LNF for two different settings of the fraction $\rho/\delta$:  narrow robustness tube, $\rho/\delta=0.5$ and wider tube, $\rho/\delta=1.15$, see Figure~\ref{fig:cube}. With increasing number of UAS, the number of collisions between initially pre-planned trajectories increase (before LNF) and the number of not collisions by LNF, while small, increases as well (figure~\ref{fig:c2}). The corresponding decay in separation rate over pairs of collisions resolved is faster for the case of $\rho/\delta=0.5$ which is expected due to less room to maneuver. Separation rate is higher when the $\rho/\delta$ ratio is higher, see Figure~\ref{fig:c1}. We performed simulations for up to 70 UAS. Average separation rate for 70 UAS is $0.915$ for $\rho/\delta=0.5$ and $0.987$ for $\rho/\delta=1.15$.  The results show that LNF can still succeed in scenarios with a high UAS density. Videos of the simulations are available at \url{https://tinyurl.com/unit-cube}.


\begin{figure}[tb]
	\begin{subfigure}[b]{0.48\columnwidth}
		\includegraphics[width=\textwidth]{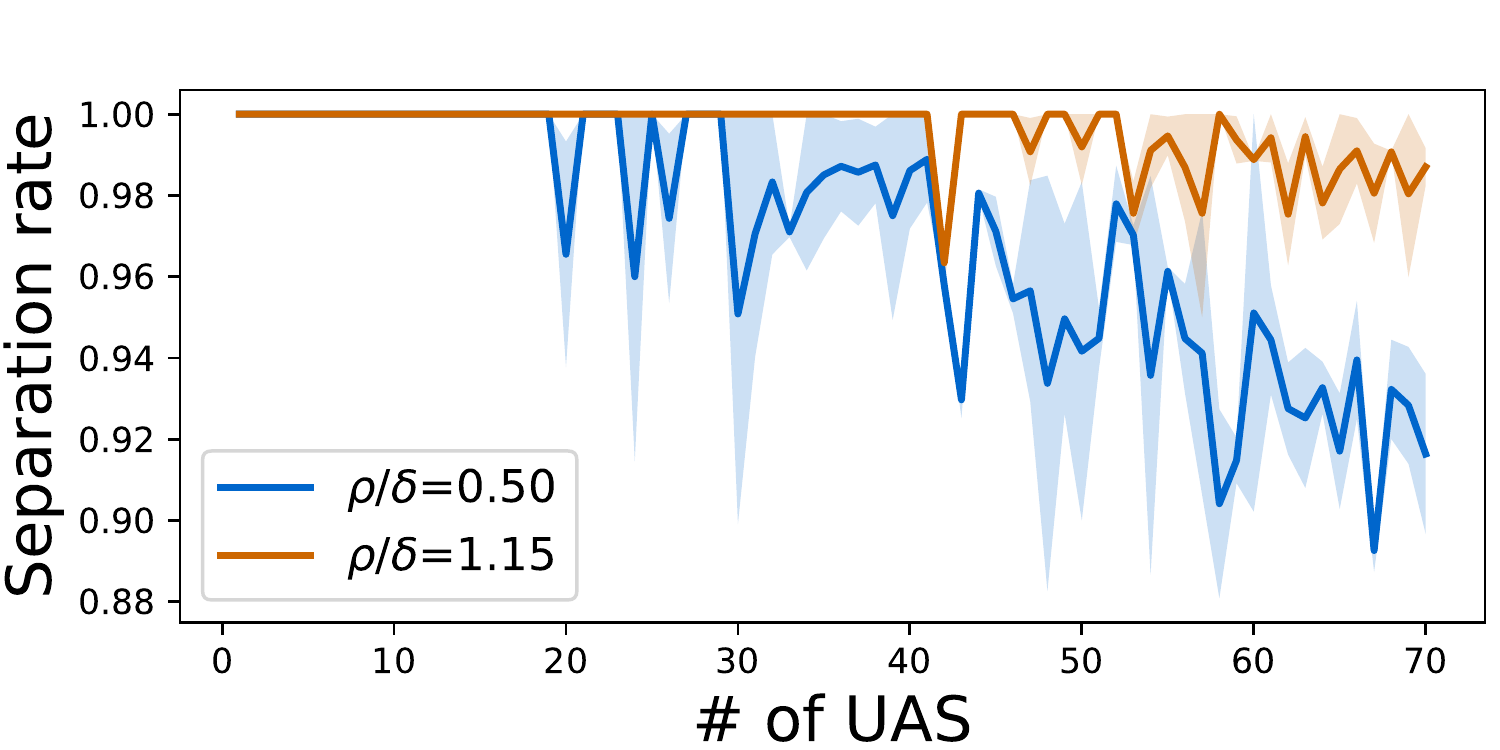}
		\caption{\textit{Separation rate} defines the fraction between the number of initial violations of the minimum separation and the number of resolved violations by LNF.}
		\label{fig:c1}
	\end{subfigure}
\hfill
	\begin{subfigure}[b]{0.48\columnwidth}
		\includegraphics[width=\textwidth]{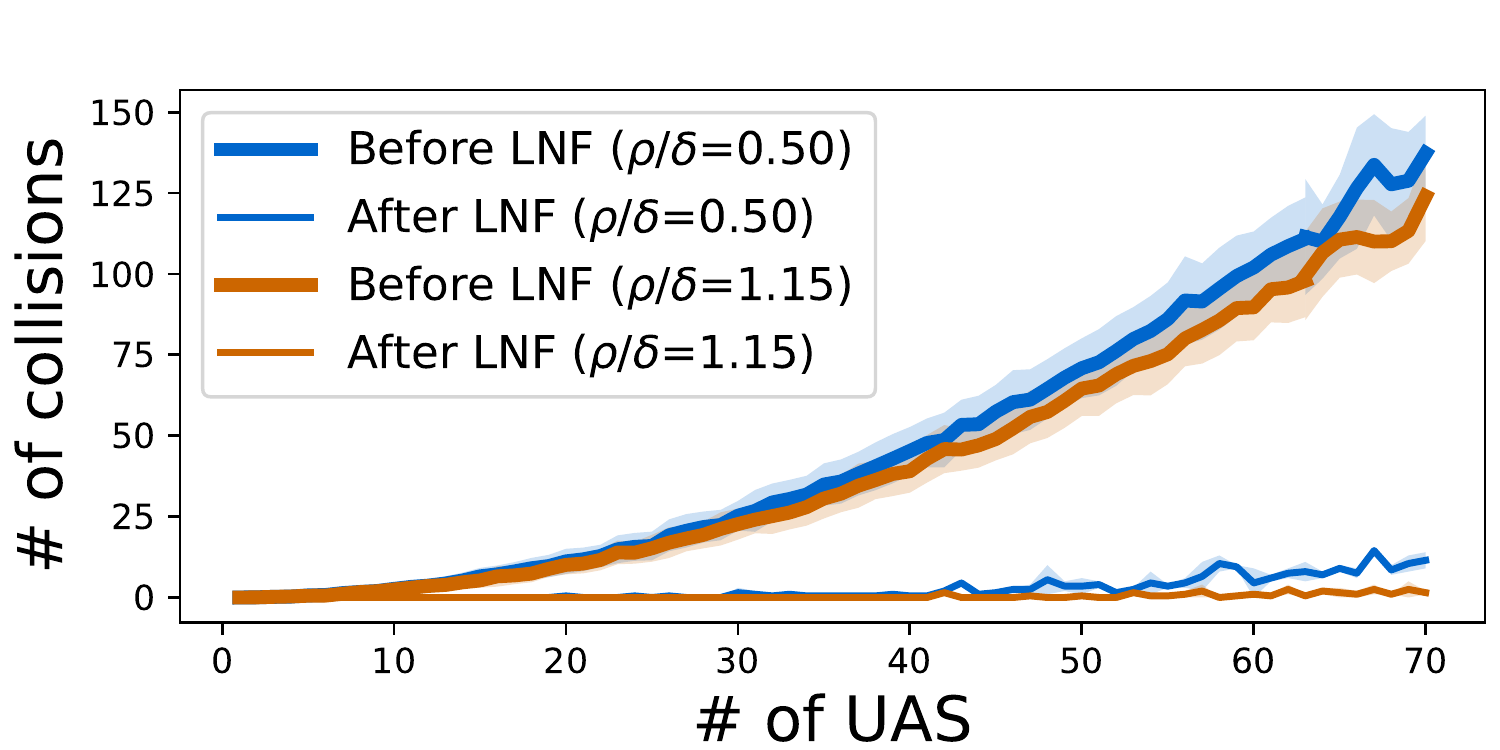}
		\caption{Number of minimum separation violations before and after LNF, averaged over 100 simulations. \\}
		\label{fig:c2}
	\end{subfigure}
	\vspace{-8pt}
	\caption{Unit cube scenario. Model performance analysis with respect to variations in the number of UAS for two different settings of $\rho/\delta$. A higher $\rho/\delta$ implies there is more room within the robustness tubes to maneuver for CA. Performance is measured in terms of separation rate (a) and the overall number of minimum separation requirement violations before and after LNF (b). We plot the mean and standard deviation over 100 iterations.}
	\label{fig:cube}
\end{figure}

\subsubsection{Comparison to MILP-based re-planning}
\label{sec:baseline}

\begin{table}[]
	\vspace{-7pt}
	\renewcommand{\arraystretch}{1.3}
	\setlength{\tabcolsep}{0.3pt}
	\centering
	\begin{tabular}{l|c||c|c|c|c|c}
		& \multicolumn{1}{c||}{Re-planning scheme} & $N=10$ & $N=20$& $N=30$ & $N=40$& $N=50$  \\ \hline\hline
		\multirow{2}{*}{\begin{tabular}[c]{@{}l@{}}Comp. times \\ (mean$\pm$std)\end{tabular}} & MILP-based planner                        
		&$0.6\pm 0.1$s 
		&$8.8\pm 9.6$s             
		&$175.5\pm\! 149.9$s            
		&$1740.\pm 129.3$s  
		& Timeout  \\ \cline{2-7} 
		& CA with LNF                 
		&15.2$\pm$ 5.1ms  &  73.1$\pm$23.5ms   
		&117.3$\pm\!$ 45.6ms   
		&  198.7$\pm$ 73.6ms
		& 211.1$\pm$82.3 ms
	\end{tabular}
	\caption{\textit{Computation times} (mean and standard deviation) demanded by the re-planning scheme (MILP-based re-planning or CA with LNF) averaged over $100$ random runs. Time taken by the MILP-based re-planner to encode the problem is not included in the overall computation time. `Timeout' stands for a timeout after $35$ minutes.}
	\label{tbl:baseline}
	\vspace{-15pt}
\end{table}

LNF requires the new trajectories after \ac{CA} to be be within the robustness tubes of pre-planned trajectories to still satisfy other high-level requirements (problem \ref{prob:deconfliction}). While this might be restrictive, we show that online re-planning is usually not an option in these multi-UAS scenarios. A MILP-based planner, similar in essence to \cite{raman2014model}, was implemented and treated as a baseline to compare against LNF through evaluations on the scenario of Section \ref{sec:case_unit_cube}. Unlike the decentralized LNF, such MILP-planner baseline is centralized as it plans for all the UAS in a single optimization to avoid the $\textit{NoFly}$-zone, reach their destinations and also avoid each other. 


We ran 100 simulations for various numbers of UAS, with each iteration having randomized start and goal positions. Simulations are available at \url{https://tinyurl.com/re-milp}. The computation times are presented in Table~\ref{tbl:baseline}. As the number of UAS increases, it is clear the online re-planning is intractable. For example, the baseline takes on average $8.8$ seconds to produce trajectories for $20$ UAS, in contrast with $73.1$ milliseconds for LNF to perform CA.
For 50 UAS and higher the MILP baseline solver could not return a single feasible solution, while LNF could. 
LNF outperforms the MILP-based re-planning baseline since it can perform \ac{CA} with small computation times, even for a high number of UAS. 


\section{Conclusions}
\label{sec:conclusions}

\textbf{Summary:} 
We presented Learning-to-Fly (L2F), an online, decentralized and mission-aware scheme for pairwise \ac{UAS} Collision Avoidance. Through Learning-And-Flying (LNF) we extended it to work for cases where more than two \ac{UAS} are on collision paths, via a systematic pairwise application of L2F and with a set-shrinking approach to avoid live-lock like situations. These frameworks combine learning-based decision-making and decentralized linear optimization-based Model Predictive Control (MPC) to perform \ac{CA}, and we also developed a fast heuristic to \textit{repair} the decisions made by the learing-based component based on the feasibility of the optimizations. Through extensive simulation, we showed that our approach has a computation time of the order of milliseconds, and can perform \ac{CA} for a wide variety of cases with a high success rate even when the UAS density in the airspace is high. 

\noindent \textbf{Limitations and future work:} 
While our approach works very well in practice, it is not \textit{complete}, i.e. does not guarantee a solution when one exists, as seen in simulation results for L2F. This drawback requires a careful analysis for obtaining the sets of initial conditions over the conflicting UAS such that our method is guaranteed to work. In future work, we aim to leverage tools from formal methods, like falsification, to get a reasonable estimate of the conditions in which our method is guaranteed to work. We will also explore improved heuristics for the set-shrinking in LNF, as well as the \ac{CR}-decision repairing procedure.


%
\bibliographystyle{ACM-Reference-Format}
\bibliography{root_ACM_TCPS}

\appendix

\section{Robustness Tubes Shrinking}
\label{sec:rts_appndx}

\begin{mydef}
	\label{def:dist}
	The distance between two sets
	$A$ and $B$ is defined as:
	\begin{equation}
	\label{eq:dist}
	\mathbf{dist}(A, B) = \inf 
	\left\{
	||a - b||_\infty \ \vert\  a\in A,\ b\in B
	\right\}
	\end{equation}
\end{mydef}

\begin{mydef}
	\label{def:tube_sep}
	Two robustness tubes $\rtube_1$ and $\rtube_2$ are said to be $\delta$-separate from each other
	if at every time step $k$ the distance between them is at least $\delta$, i.e.
	\begin{equation}
	\label{eq:tube_deltasep}
	\mathbf{dist}(P_{1,k}, P_{2,k}) \geq \delta\ \forall k=0,\ldots,H.
	\end{equation}	
	For brevity we use $\mathbf{dist}(\rtube_1, \rtube_2)\geq \delta$ for denoting being $\delta$-separate across all time indices $k=0,\ldots,H$.
\end{mydef}

\begin{proof}[Proof of Theorem~\ref{th:rts_success}]
	By construction of RTS, see Algorithm~\ref{alg:shrinking}. 
	If initial tubes are $\delta$-separate then no shrinking is required and therefore, both properties \eqref{eq:lnf_tube_subset} and \eqref{eq:lnf_tube_dist} are satisfied. 
	If the initial tubes are not $\delta$-separate, property~\eqref{eq:lnf_tube_subset} comes from the fact that for any time step $k$, $P_{j, k}'=P_{j, k} \setminus \varPi_k$ for UAS $j=1,2$. 
	Property~\eqref{eq:lnf_tube_dist} is a consequence of the zero-slack solution and Theorem~\ref{th:CAMPC_success} which states that resulting trajectories are non-conflicting, $||p_{1, k}'-p_{2, k}'||\geq\delta$, $\forall k\in\{0,\ldots,H\}$, therefore, $msep\geq\delta$. Following Algorithm~\ref{alg:shrinking}, for any time step $k$ box's $\varPi_k$ smallest edge is $\min(msep, \delta)=\delta$ and since for both UAS $j=1,2$ the tubes update is defined as
	$P_{j, k}'=P_{j, k} \setminus \varPi_k$, the shrinked tubes $P_{j, k}'$ are $\delta$-separate. 
\end{proof}

\begin{proof}[Proof of Lemma~\ref{lemma:p_in_tube}]
	From the CA-MPC definition~\eqref{eq:campc} it follows that $\mathbf{p}_j'\in\rtube_j$, $\forall j\in\{1,2\}$. The updated tubes are defined as $\rtube_j' = \rtube_j \setminus \boldsymbol\varPi$, see Algorithm~\ref{alg:shrinking}. By the definition of 3D cube $\boldsymbol\varPi$, for any time step $k$, ${p}_{j,k}'\not\in\varPi_k$, therefore, $\mathbf{p}_j' \in \rtube_j',\ \forall j\in\{1,2\}$.
\end{proof}

\begin{proof}[Proof of Lemma~\ref{lemma:tubes}]
	Following the Definition~\ref{def:tube_sep}, tubes are $\delta$-separate if $\mathbf{dist}(P_{1,k}, P_{2,k}) \geq \delta,\ \forall k\in\{0,\ldots,H\}$.
	Therefore, due to~\eqref{eq:dist} the following holds:
	\begin{equation}
	\inf \left\{
	||p_{1,k} - p_{2,k}|| \mid\  p_{1,k}\in P_{1,k},\ p_{2,k}\in P_{2,k}\right\}\geq \delta.
	\end{equation}
	By the definition of the infimum operator, $\forall p_{1,k}\in P_{1,k}, \forall p_{2,k}\in P_{2,k}$:
	\begin{equation}
	||p_{1,k} - p_{2,k}|| \geq 
	\inf \left\{
	||p_{1,k} - p_{2,k}|| \mid\  p_{1,k}\in P_{1,k},\ p_{2,k}\in P_{2,k}\right\}\geq \delta,
	\end{equation}
	which completes the proof.
\end{proof}

\begin{proof}[Proof of Lemma~\ref{lemma:3uas}]
	\begin{enumerate}
		\item Property~\eqref{eq:ms23} directly follows from Theorem~\ref{th:CAMPC_success}.
		\item Due to Theorem~\ref{th:rts_success}, RTS application  \eqref{eq:prop13} leads to tubes $\rtube''_1$ and $\rtube'_3$ being $\delta$-separate. RTS \eqref{eq:prop23} leads to $\rtube''_3\subseteq \rtube'_3$. Therefore, $\rtube''_1$ and $\rtube''_3$ are $\delta$-separate and following Lemma~\ref{lemma:tubes}, property~\eqref{eq:ms13} holds.
		\item Analogously, due to Theorem~\ref{th:rts_success}, RTS application \eqref{eq:prop12} leads to tubes $\rtube'_1$ and $\rtube'_2$ being $\delta$-separate. 
		RTS \eqref{eq:prop13} leads to
		$\rtube''_1\subseteq \rtube'_1$ and RTS \eqref{eq:prop23} leads to
		$\rtube''_2\subseteq\rtube'_2$. Therefore, $\rtube''_1$ and $\rtube''_2$ are $\delta$-separate and following Lemma~\ref{lemma:tubes}, property~\eqref{eq:ms12} holds.
		\item Tube belonging property \eqref{eq:im} follows directly from Lemma~\ref{lemma:p_in_tube}.
	\end{enumerate}
\end{proof}

\section{Links to the videos}

Table~\ref{tbl:sims} has the links for the
visualizations of all simulations and experiments performed in
this work.

\begin{table}[h!]
	\vspace{-10pt}
	\renewcommand{\arraystretch}{1.}
	\setlength{\tabcolsep}{2.5pt}
	\centering
	\begin{tabular}{c|c|c|c|c}
		\textbf{Scenario} & \textbf{Section} & \textbf{Platform} & \textbf{$\#$ of UAS} & \textbf{Link} \\ \hline
		L2F test & Sec.~\ref{sec:experiments_l2f} & Sim. & 2 & \url{https://tinyurl.com/l2f-exmpl} \\ \hline
		Crazyflie validation & Sec.~\ref{sec:experiments_l2f} & CF 2.0 & 2 & \url{https://tinyurl.com/exp-cf2} \\ \hline
		Four UAS position swap
		& Sec.~\ref{sec:case_pos_swap} & Sim. & 4& 
		\url{https://tinyurl.com/swap-pos} \\ \hline
		Four UAS reach-avoid mission & 
		Sec.\ref{sec:case_reach_avoid} & Sim. & 4&
		\url{https://tinyurl.com/reach-av} \\ \hline
		High-density unit cube & 
		Sec.\ref{sec:case_unit_cube} & Sim. & 10, 20, 40 &
		\url{https://tinyurl.com/unit-cube}\\ \hline
		MILP re-planning & Sec~\ref{sec:baseline} & MATLAB & 20 &
		\url{https://tinyurl.com/re-milp}
	\end{tabular}
	\caption{Links for the videos for simulations and experiments. ``Sim.'' stands for Python simulations, ``CF2.0'' for experiments on the Crazyflies.}
	\label{tbl:sims}
\end{table}

\end{document}